\documentclass[11pt,letterpaper]{article}

\usepackage[margin=1in]{geometry}

\usepackage{setspace}
\usepackage{parskip}

\usepackage[T1]{fontenc}
\usepackage[utf8]{inputenc}
\usepackage{lmodern} 

\usepackage{microtype}

\usepackage{amsmath,amssymb,amsthm}

\usepackage{graphicx}
\usepackage{booktabs}
\usepackage{graphicx}
\usepackage{latexsym}
\usepackage{times}
\usepackage{soul}
\usepackage{url}
\usepackage[utf8]{inputenc}
\usepackage[small]{caption}
\usepackage{booktabs}
\usepackage{algorithm}
\usepackage{comment}
\usepackage[switch]{lineno}
\usepackage{tikz}
\usetikzlibrary{graphs}
\usetikzlibrary{arrows.meta, positioning}
\usetikzlibrary{patterns}
\usepackage{algpseudocode}
\usepackage{luatex85}
\usepackage{subcaption}
\usepackage{bussproofs}

\def\Underline{\setbox0\hbox\bgroup\let\\\endUnderline}
\def\endUnderline{\vphantom{y}\egroup\smash{\underline{\box0}}\\}
\def\|{\verb|}

\theoremstyle{plain}
\newtheorem{definition}{Definition}[section]

\newtheorem{lemma}[definition]{Lemma}
\newtheorem{prop}[definition]{Proposition}

\newtheorem{definition-prop}[definition]{Definition-Proposition}

\newtheorem{example}{Example}
\newtheorem{theorem}{Theorem}

\usepackage[hidelinks]{hyperref}
\usepackage{cleveref}

\crefalias{prop}{prop}
\crefname{prop}{Proposition}{Proposition}
\crefname{appendix}{Appendix}{Appendices}

\usepackage{enumitem}
\usepackage{comment}
\usepackage{lineno}
\usepackage[page]{appendix}
\usepackage{authblk}

\newcommand{\cvle}{\mathop{\mathrm{cvle}}}
\newcommand{\cvri}{\mathop{\mathrm{cvri}}}
\newcommand{\cover}{\mathop{\mathrm{cover}}}

\newcommand{\ini}{w_{0}}

\date{} 

\title{Arrow-Type Impossibility for Genuinely Modal Judgments}

\author[1]{Yutaka Nagai\thanks{nagai.yutaka.u8@s.mail.nagoya-u.ac.jp}}
\author[1]{Hirotaka Ono\thanks{ono@i.nagoya-u.ac.jp}}
\affil[1]{Graduate School of Informatics, Nagoya University, Furo-cho, Chikusa-ku, Nagoya 464--8601, Japan}

\begin{document}

\maketitle

\begin{abstract}
Judgment aggregation studies how to combine individual judgments on logically related propositions into a collective judgment.
Classical impossibility results show that sufficiently strong logical interconnections force dictatorship under natural aggregation axioms.
In this paper, we ask whether such impossibility can still arise when the objects of aggregation are required to be genuinely modal judgments rather than plain factual propositions.
Since modal logic contains propositional logic, this question is meaningful only if one excludes fact-based aggregation in disguise.
We show that Arrow-type impossibility already re-emerges in a strikingly sparse modal setting.
We prove an impossibility theorem on a simple cyclic frame for an agenda generated from a single propositional variable by repeated applications of a single modal operator, and we further demonstrate this phenomenon for an alternative family of frames satisfying a natural symmetry condition.
Thus, even under a modal-operator requirement, semantic structure alone can generate the logical interconnections needed for dictatorship.
Technically, our analysis has two layers.
First, we prove a semantic reduction theorem showing that certain iterated modal patterns can be collapsed by shifting the evaluation point.
Second, building on this reduction, we identify a local-to-global frame mechanism by which frame geometry yields minimally inconsistent modal judgment sets and the strong path-connectivity required for impossibility.
The same reduction also turns consistency checking into a small combinatorial covering problem, which yields efficient implementations of non-dictatorial aggregation procedures.
\end{abstract}

\begin{quote}
\begin{small}
    \end{small}
\end{quote}
\noindent

\section{Introduction}

How should we aggregate rational judgments?  
Judgment aggregation studies how to combine individual judgments on logically related propositions into a collective judgment \cite{Listdicur}.  
A central difficulty is that logical interconnections among propositions may turn individually rational judgments into an irrational collective outcome.  
A classical example is the doctrinal paradox, where majority voting on $p$, $q$, and $p \land q$ yields an inconsistent collective judgment even though each individual judgment set is consistent \cite{Kornhauser1986, Pettit}.  
This phenomenon has led to a rich line of impossibility results in judgment aggregation \cite{Dietrich, List, Dokow2010, Listdicur}. In particular, subsequent work showed that sufficiently strong logical interconnections force dictatorship under natural aggregation axioms \cite{List, Dokow2010}.

These impossibility results are by now understood from several strong viewpoints. At the level of logical frameworks, propositional judgment aggregation and its abstractions have been extensively studied \cite{Dietrich, List}. 
At the level of agendas and feasible evaluations, previous work gives sharp structural characterizations of when dictatorship is unavoidable \cite{Dokow2010}. Thus, the classical picture is already powerful: impossibility can be identified either through the logical structure of the agenda or through the geometry of feasible $0$--$1$ evaluations.

Our starting point is a different question.  
Much of the existing literature allows the aggregation of judgments on plain factual propositions themselves.  
This is natural and has led to powerful impossibility results.  
However, many aggregation problems are not primarily about settling bare factual truth, but about aggregating attitudes toward facts under uncertainty or interpretation: whether a claim must hold, may hold, or cannot hold.  
These are inherently modal judgments.  
This suggests studying judgment aggregation in a setting where the aggregated items are themselves modal.

At this point, however, a basic difficulty arises.  
Modal logic contains propositional logic, so classical impossibility results remain present in a trivial way if one simply allows modal formulas in addition to propositional ones.  
That is not the phenomenon we want to isolate.  
Our goal is to set aside such fact-based aggregation and ask what happens when every aggregated item is required to involve modal operators.  
In other words, we ask whether Arrow-type impossibility still arises when the objects of aggregation are genuinely modal judgments rather than propositional judgments in disguise.

At first sight, one might expect impossibility to become harder to obtain in such a setting. Modal operators seem to loosen, rather than rigidify, logical structure: they move attention from outright truth to possibility, necessity, or related attitudes. One may therefore suspect that, once fact-based aggregation is excluded in the above sense, purely modal agendas would no longer generate the strong logical interconnections needed for dictatorship.
However, the situation is more subtle. Unlike propositional formulas, modal formulas are governed not only by syntax but also by the geometry of the underlying frame. 

Our main result shows that this semantic structure itself can generate the interconnections needed for impossibility.
We prove an Arrow-type impossibility theorem on a simple cyclic frame, for an agenda generated from a single propositional variable by repeated applications of a single modal operator.  
We further show that the same phenomenon arises for an alternative family of frames satisfying a natural symmetry condition under another restricted iterated modal pattern. 
Thus, even under a modal-operator requirement, impossibility already reappears in a strikingly sparse setting.

The technical core of our analysis has two layers.  
First, we prove a semantic reduction theorem showing that certain iterated modal patterns can be collapsed by shifting the evaluation point.  
This provides a semantic normalization of modal formulas and reduces consistency questions to a finite combinatorial form.  
Second, building on this reduction, we identify a local overlap pattern among reachable regions on the frame, which in turn yields minimally inconsistent modal judgment sets.  
Combined with a global propagation argument based on a coprimality condition, this also yields the strong path-connectivity required for impossibility.  
Thus, the proof identifies a concrete route from frame geometry to logical interconnection, and from there to dictatorship.  
The same reduction also turns consistency checking into a small combinatorial covering problem, which in turn yields efficient implementations of non-dictatorial aggregation procedures. 

Our work complements existing impossibility results in judgment aggregation.  
Prior work gives powerful characterizations of when a given agenda, or more abstractly a feasible set of binary evaluations, forces dictatorship under natural axioms \cite{List, Dokow2010}.  
In particular, binary-aggregation results already identify sharp agenda-level sources of impossibility \cite{Dokow2010}.  
Our focus is different: we ask how such impossibility-inducing interconnections can be generated by the logic itself once the objects of aggregation are required to be genuinely modal judgments.  
This viewpoint is specific to modal logic.  
While modal logic has appeared in the literature as a meta-language for describing aggregation procedures \cite{Novaro, Pauly, Agotnes}, here modal propositions themselves are the objects of aggregation.

\medskip

\noindent\textbf{Our contributions.}
\begin{enumerate}
\item\label{con1} \textbf{Semantic reduction of modal operators.}
We prove a semantic reduction theorem for the symmetric cyclic setting, showing that certain iterated modal patterns can be collapsed by shifting the evaluation point.  
This gives a semantic normalization of modal formulas and serves as the main technical tool throughout the paper.
(Theorem~\ref{th:calcu reduc pro})

\item\label{con2}\textbf{A local-to-global frame mechanism for impossibility.}
Building on the reduction, we identify a local frame pattern that yields minimally inconsistent modal judgment sets, and thus a one-step propagation relation.
Combined with a coprimality argument, this propagates to the strong path-connectivity needed for impossibility.
This gives a concrete route from frame geometry to the logical interconnections that force dictatorship.

\item\label{con3}\textbf{Arrow-type impossibility in a strikingly sparse modal setting.}
We prove that impossibility already arises on a simple cyclic frame for an agenda generated from a single propositional variable by repeated applications of a single modal operator.  
Thus, even after setting aside fact-based aggregation in the sense discussed above, dictatorship can already be forced by a highly restricted modal pattern.  
We further extend this phenomenon to a broader family of frames satisfying a natural symmetry condition.
(Theorem~\ref{th:impossible})
 \item \label{con4} \textbf{Efficient implementation of non-dictatorial aggregation procedures.}
    We apply the sequential majority procedure of Peleg and Zamir \cite{Peleg} to our modal setting.  
    The reduction turns consistency checking, which is the bottleneck of the procedure, into a small combinatorial covering problem.  
    This yields efficient algorithms for implementing non-dictatorial aggregation procedures in our setting.
    (Section~\ref{sec:altaggre})
\end{enumerate}
Taken together, our results show that Arrow-type impossibility does not disappear when aggregation is restricted to genuinely modal judgments.  
Even in a highly sparse modal setting, semantic structure alone already generates enough logical interconnection to force dictatorship.  
In this sense, our results suggest that Arrow-type impossibility is not confined to rigid propositional agendas, but can re-emerge directly from modal semantics itself.
At the same time, the same structural analysis yields efficient implementations of non-dictatorial aggregation procedures, linking impossibility and algorithmic tractability through a common semantic reduction.
	

\section{Preliminaries}\label{sec:pre modal}
In Subsection \ref{subsec:pre modal}, we define the syntax and semantics of modal logic. In Subsection \ref{subsec:JA}, we present the setting of judgment aggregation. 

\subsection{Modal Logic}\label{subsec:pre modal}
In this subsection, we introduce the restricted modal language used in the paper, recall its Kripke semantics, and then specify the two frame classes on which our analysis is carried out.
\paragraph{A minimal modal language.}
Modal logic enriches propositional statements with operators such as \emph{necessity} ($\Box$) and \emph{possibility} ($\Diamond$) \cite[Section 1.2]{Blackburn}.
Intuitively, $\Box p$ means that $p$ holds necessarily, while $\Diamond p$ means that $p$ may hold.
For example, if $p$ denotes the statement ``Alice lives in New York,'' then $\Box p$ means that Alice necessarily lives in New York, while $\Diamond p$ means that Alice may live in New York.
In this paper, we work with a minimal modal language consisting of a single propositional variable $p$ and the two modal operators $\Box$ and $\Diamond$.
Thus, modal formulas are strings of the form $\alpha_1 \alpha_2 \cdots \alpha_n p$, where $n \ge 0$ and each $\alpha_i \in \{\Box,\Diamond\}$.
For example, $p$, $\Box p$, $\Diamond p$, $\Box\Box p$, $\Box\Diamond p$, $\Diamond\Box p$, and $\Diamond\Diamond p$ are modal formulas.



\paragraph{Kripke semantics.}
To interpret modal formulas, we use Kripke semantics, in which formulas are evaluated at worlds connected by an accessibility relation.
Intuitively, a Kripke frame specifies which worlds are available from which other worlds, and a valuation specifies at which worlds the propositional variable $p$ is true.
Formally, a Kripke frame and a Kripke model are defined as follows \cite[Section 1.3]{Blackburn}.

\begin{enumerate}
	\item $W$: a non-empty set of possible worlds.
	\item $R \subseteq W \times W$: an accessibility relation. $(w_{1}, w_{2}) \in R$ means that the world $w_{1}$ can access $w_{2}$.
	\item $V \subseteq W$: a valuation.
\end{enumerate}

	The pair $(W, R)$ and the triple $M=(W, R, V)$ are called a Kripke frame and a Kripke model, respectively. For all worlds $w$ and modal formulas $\Phi$, we define a satisfaction relation $M \vDash w : \Phi$ inductively. $M \vDash w : \Phi$ means that $\Phi$ is true at $w$ under the Kripke model $M$. The notation $w: \Phi$ follows \cite{Negri2005}.
\begin{align*}
	M \vDash w : p  \,\, \qquad &\Leftrightarrow \quad w \in V \\
	M \vDash w_{1} : \Box \Phi  \quad &\Leftrightarrow \quad \text{For all $w_{2} \in W$, $(w_{1}, w_{2}) \in R$ implies $M \vDash w_{2} : \Phi$.} \\
	M \vDash w_{1} : \Diamond \Phi  \quad &\Leftrightarrow \quad \text{There exists some $w_{2} \in W$ such that $(w_{1}, w_{2}) \in R$ and $M \vDash w_{2} : \Phi$.}
\end{align*}

\Cref{fig:Kripke frame} illustrates a Kripke frame with $W=\{0,1,2,3\}, 
R=\{(0,1),(0,2),(1,1),(1,3),(2,3),(3,0)\}$. Let the valuation be $V=\{0,1\}$.
Since the worlds accessible from $0$ are $1$ and $2$, we have
\[
M \vDash 0 : \Box p
\Leftrightarrow
M \vDash 1 : p \land M \vDash 2 : p
\Leftrightarrow
1 \in V \land 2 \in V.
\]
By $2 \notin V$, it follows that $M \nvDash 0 : \Box p$.
Similarly, since the worlds accessible from $1$ are $1$ and $3$, we have
\[
M \vDash 1 : \Diamond p
\Leftrightarrow
M \vDash 1 : p \lor M \vDash 3 : p
\Leftrightarrow
1 \in V \lor 3 \in V.
\]
By $1 \in V$, we have $M \vDash 1 : \Diamond p$.

\begin{figure}[htbp]
  \begin{minipage}[b]{0.30\columnwidth}
    \centering
    \centering
\begin{tikzpicture}[
    every node/.style={circle, draw, minimum size=5mm}, scale=0.7,
    every edge/.style={draw, ->}
]

	\node (0) at (0, 0) {0};
    	\node (1) at (2, 0) {1};
	\node (2) at (0, 2) {2};
	\node (3) at (2, 2) {3};
	
    	\path (0) edge (1);
    	\path (1) edge (3);
    	\path (0) edge (2);
    	\path (2) edge (3);
    	\path (3) edge (0);
    	\path (1) edge[loop right] (1);
    \path (0, -2.0);
\end{tikzpicture}
    \caption{A Kripke frame}
    \label{fig:Kripke frame}
  \end{minipage}
  \hfill 
\begin{minipage}[b]{0.65\columnwidth}
    \centering
\begin{subfigure}{0.48\textwidth}
\centering
\begin{tikzpicture}[
    every node/.style={circle, draw, minimum size=5mm},
    every edge/.style={draw, ->},
    scale=0.7,
    transform shape
]

	\node (-1) at (1.73, 5) {x};
	\node (0) at (0,4) {0};
	\node (1) at (1.73,3) {1};
	\node (2) at (1.73,1) {2};
	\node (3) at (0,0) {3};
	\node (4) at (-1.73,1) {4};
	\node (5) at (-1.73,3) {5};
	
	\path (-1) edge (0);
	\path (-1) edge (1);
	\path (0) edge (1);
	\path (1) edge (2);
	\path (2) edge (3);
	\path (3) edge (4);
	\path (4) edge (5);
	\path (5) edge (0);

\end{tikzpicture}
\caption{Frame 1}
\label{fig:GALS0}
\end{subfigure}
\hfill
\begin{subfigure}{0.48\textwidth}
\centering
\begin{tikzpicture}[
    every node/.style={circle, draw, minimum size=5mm},
    every edge/.style={draw, ->},
    scale=0.7,
    transform shape
]

	\node (0) at (0,4) {0};
	\node (1) at (1.73,3) {1};
	\node (2) at (1.73,1) {2};
	\node (3) at (0,0) {3};
	\node (4) at (-1.73,1) {4};
	\node (5) at (-1.73,3) {5};
	
	\path (0) edge (1);
	\path (0) edge (2);
	\path (1) edge (2);
	\path (1) edge (3);
	\path (2) edge (3);
	\path (2) edge (4);
	\path (3) edge (4);
	\path (3) edge (5);
	\path (4) edge (5);
	\path (4) edge (0);
	\path (5) edge (0);
	\path (5) edge (1);

	\path (0) edge[loop above] (0);
	\path (1) edge[loop right] (1);
	\path (2) edge[loop right] (2);
	\path (3) edge[loop below] (3);
	\path (4) edge[loop left] (4);
	\path (5) edge[loop left] (5);

\end{tikzpicture}
\caption{Frame 2}
\label{fig:GALS}
\end{subfigure}
   \caption{Examples of Frame 1 and 2}
\label{fig:frames}
  \end{minipage}
\end{figure}

\paragraph{Frame classes used in this paper.}
Fix positive integers $r>k\ge 1$ and a nonempty subset $A \subseteq \mathbb{Z}/r\mathbb{Z}$.
We refer to Appendix \ref{appendix: addition} for the definition of $\mathbb{Z}/r\mathbb{Z}$ and its addition $+$.

We consider the following two classes of Kripke frames.
\begin{enumerate}
	\item Frame 1 : $W:=\{x\} \sqcup \mathbb{Z}/r\mathbb{Z}, R:=\{(x, a) \mid a \in A\} \sqcup \{(w, w+k) \mid w \in \mathbb{Z}/r\mathbb{Z}\}$.
	\item Frame 2 : $W:=\mathbb{Z}/r\mathbb{Z}, R:=\{(w, w+a) \mid w \in \mathbb{Z}/r\mathbb{Z}, a \in A\}$.
\end{enumerate}

For simplicity, 
we usually abbreviate the residue class $\overline{a} \in \mathbb{Z}/r\mathbb{Z}$ as $a$. 
\Cref{fig:frames} shows examples of Frame~1 and Frame~2, with parameters $(r,k,A)=(6,1,\{0,1\})$ and $(6,2,\{0,1,2\})$, respectively.

These frame classes are related to cyclic pursuit.
In particular, if $r$ agents are arranged on a ring and each agent $i$ tracks agent $i+1 \bmod r$, then the resulting reference graph is Frame~2 with $W=\mathbb{Z}/r\mathbb{Z}$ and $A=\{1\}$; see \cite{Klamkin1971, Marshall2004, Park2024, Parsegov2023}.

\subsection{Judgment aggregation}\label{subsec:JA}
In this subsection, we introduce the judgment-aggregation setting for our modal agendas.
We first define rational judgment sets and aggregation rules, and then introduce the logical-interconnection notions used later in the definition of impossibility frames.

Let $N=\{1,2,\cdots,n\}$ with $n\ge 2$ be the set of individuals.
For every modal formula $\Phi$, we write $\lnot \Phi$ for its negation, where
$M \vDash w:\lnot \Phi \Leftrightarrow M \nvDash w:\Phi$.
Throughout the paper, we identify $\lnot\lnot\Phi$ with $\Phi$.

\paragraph{Agenda and rational judgment sets.}
Fix a set of modal formulas $\Gamma$, called an \textit{agenda}.
In this paper, we consider the agendas $\Gamma=\{\Box^{j}p,\lnot\Box^{j}p\}_{j\ge 1}$ or $\{(\Box\Diamond)^{j}\Box p,\lnot(\Box\Diamond)^{j}\Box p\}_{j\ge 0}$.
Although each of these sets is infinite syntactically, it has only finitely many semantic equivalence classes, so for the purposes of our analysis the agenda can be treated as finite.

Fix a Kripke frame $(W,R)$, which is either Frame~1 or Frame~2, and a world $\ini\in W$.
A subset $\Gamma_0\subseteq \Gamma$ is \textit{complete} in $\Gamma$ if for every $\Phi\in\Gamma$, either $\Phi\in\Gamma_0$ or $\lnot\Phi\in\Gamma_0$.
A subset $\Gamma_0\subseteq \Gamma$ is \textit{consistent} under $(W,R)$ if there exists a valuation $V\subseteq W$ such that
$(W,R,V)\vDash \ini:\Phi$ for all $\Phi\in\Gamma_0.$ 
A subset $\Gamma_0\subseteq \Gamma$ is \textit{rational} in $\Gamma$ if it is complete in $\Gamma$ and consistent under $(W,R)$.

We write $\mathcal J$ for the set of all rational subsets of $\Gamma$.
A subset $\Delta\subseteq \Gamma$ is \textit{minimally inconsistent} under $(W,R)$ if for every $\Delta_0\subseteq \Delta$, the set $\Delta_0$ is inconsistent under $(W,R)$ if and only if $\Delta_0=\Delta$.

\paragraph{Profiles, aggregation rules, and axioms.}
A \textit{profile} of individual judgments is a tuple $(\Gamma_1,\Gamma_2,\cdots,\Gamma_n)\in \mathcal{J}^N$.
We identify this tuple with the function $f:N\to\mathcal{J}$ given by $f(i)=\Gamma_i$.
For $f\in\mathcal{J}^N$ and $\Phi\in\Gamma$, we write $f^{-1}(\Phi):=\{i\in N \mid \Phi\in f(i)\}$, that is, the set of individuals who accept $\Phi$ under the profile $f$.
A \textit{(judgment) aggregation rule} is a function $F:\mathcal{J}^N\to\mathcal{J}$.
For $f\in\mathcal{J}^N$ and $\Phi\in\Gamma$, the statement $\Phi\in F(f)$ means that the group accepts $\Phi$ under the profile $f$.

We impose the following axioms on aggregation rules.
\begin{description}[leftmargin=2em, labelindent=1em]
    \item[\textbf{Unanimity \cite{List}.}]
    For all $f\in\mathcal{J}^N$ and $\Phi\in\Gamma$, if $f^{-1}(\Phi)=N$, then $\Phi\in F(f)$.
    That is, if all individuals accept $\Phi$, then the group also accepts $\Phi$.

    \item[\textbf{Independence \cite{List}.}]
    For all $f,g\in\mathcal{J}^N$ and $\Phi\in\Gamma$, if $f^{-1}(\Phi)=g^{-1}(\Phi)$, then $\Phi\in F(f) \;\Leftrightarrow\; \Phi\in F(g).$
    That is, the collective acceptance of $\Phi$ depends only on which individuals accept $\Phi$.

    \item[\textbf{Positive-Negative Neutrality.}]
    For all $f,g\in\mathcal{J}^N$ and $\Phi\in\Gamma$, if $f^{-1}(\Phi)=g^{-1}(\lnot\Phi)$, then $   \Phi\in F(f) \;\Leftrightarrow\; \lnot\Phi\in F(g).$
 
    That is, acceptance of $\Phi$ and acceptance of $\lnot\Phi$ are treated symmetrically across profiles.
    This axiom generalizes Neutrality in \cite{May}.

    \item[\textbf{Dictatorship \cite{List}.}]
    There exists an individual $i_0\in N$ such that for all $f\in\mathcal{J}^N$, $F(f)=f(i_0).$
    That is, one individual always determines the group judgment.
\end{description}

\paragraph{Agenda interconnections and impossibility frames.}
We now introduce the logical-interconnection notions used later in our impossibility results.
Following the judgment-aggregation literature, minimal inconsistency and path-connectedness capture the kind of logical dependence that can force dictatorship \cite{List}.
The agenda $\Gamma$ is \textit{minimally connected} under $(W,R)$ if there exists a subset $\Gamma_0\subseteq\Gamma$ with $|\Gamma_0|\ge 3$ such that $\Gamma_0$ is minimally inconsistent under $(W,R)$.
For formulas $\Phi,\Psi\in\Gamma$ with no negation operator, we write $\Phi <_0 \Psi$ if there exists a subset $\Gamma_0\subseteq\Gamma$ such that $\Gamma_0\cup\{\Phi,\lnot\Psi\}$
is minimally inconsistent under $(W,R)$, while $\Gamma_0\cup\{\lnot\Phi,\Psi\}$
is consistent under $(W,R)$.
The relation $<_0$ is a special case of the even-negation property in \cite{List}.
The agenda $\Gamma$ is \textit{strongly path-connected} under $(W,R)$ if for all formulas $\Phi,\Psi\in\Gamma$ with no negation operator, there exist
$\Phi=\Phi_0,\Phi_1,\cdots,\Phi_{m-1},\Phi_m=\Psi$ such that $\Phi_i <_0 \Phi_{i+1}$ for all $0\le i\le m-1$.

Finally, a Kripke frame $(W,R)$ is an \textit{impossibility frame} under $\Gamma$ if $\Gamma$ is minimally connected and strongly path-connected under $(W,R)$.
This notion is analogous to strong connectedness in \cite{List}, but there is an important difference: strong connectedness in \cite{List} is a condition on agendas, whereas our impossibility-frame notion is formulated as a condition on Kripke frames.

\section{A framework for impossibility}
In this section, we develop the framework that underlies our impossibility results.
Our route has two components.
First, we prove a semantic reduction theorem showing that certain iterated modal patterns can be collapsed by shifting the evaluation point.
Second, we introduce a general frame-based criterion under which sufficiently strong logical interconnections force dictatorship.
Together, these two ingredients reduce the main task in later sections to constructing Kripke frames that generate the required interconnections.

\subsection{Semantic reduction of modal operators}\label{sec:semantic_reduction}


We begin with the semantic reduction that forms the first component of our framework.
A central difficulty in our setting is that modal formulas are not controlled purely by syntax: their truth depends on how evaluation propagates across worlds in the underlying frame.
The key observation is that, under a natural symmetry condition, certain iterated modal blocks can be collapsed semantically by shifting the evaluation point.
This reduction will be used in two ways later: it turns modal judgments into a finite combinatorial form for the impossibility proof, and it also underlies the aggregation algorithms.

\begin{definition}
	The set $A$ is $k$-symmetric if for all $a \in A$, $k-a \in A$.
\end{definition}

\begin{theorem}\label{th:calcu reduc pro}
	Suppose that $(W, R)$ is Frame 2. If the set $A$ is $k$-symmetric, then for all worlds $w \in \mathbb{Z}/r\mathbb{Z}$, modal formulas $\Phi$, and valuations $V$, $M=(\mathbb{Z}/r\mathbb{Z}, R, V) \vDash w : \Box \Diamond \Box \Phi$ if and only if $M=(\mathbb{Z}/r\mathbb{Z}, R, V) \vDash w+k : \Box \Phi$.
\end{theorem}

\Cref{th:calcu reduc pro} is not a purely syntactic rewriting rule.  
It exploits the geometry of the frame to replace a nontrivial modal block by a shifted evaluation point, which is what makes the later reduction to set covering possible. For this reason, it also isolates a semantic phenomenon that is not visible at the purely syntactic level. The proof is provided in Appendix \ref{sec:reduc modal}.

\subsection{A Frame-based Criterion for Dictatorship}

We now introduce the second component of our framework: a general criterion that derives dictatorship from sufficiently strong logical interconnections on a fixed Kripke frame.
This criterion plays the same structural role as related dictatorship criteria in judgment aggregation, but is formulated here for our modal setting.
We prove this theorem in Appendix~\ref{appendix: arrow}.
Since Positive-Negative Neutrality implies Independence (Lemma~\ref{lem:PN imp I}), we do not assume Independence explicitly in \Cref{th:P and DN is D}.
\begin{theorem}\label{th:P and DN is D}
	Suppose a Kripke frame $(W, R)$ is an \textit{impossibility frame} under $\Gamma$. If an aggregation rule $F:\mathcal{J}^{N} \to \mathcal{J}$ satisfies Unanimity and Positive-Negative Neutrality, then it is dictatorial.
\end{theorem}
Given \Cref{th:P and DN is D}, the remaining task is to construct impossibility frames.
In the next section, the semantic reduction from \Cref{sec:semantic_reduction} will allow us to do so by converting modal consistency into a finite combinatorial problem.

\section{Impossibility frames with a propositional variable and modal operators}\label{sec:impossible agenda}
In this section, we construct impossibility frames for the modal agendas considered in this paper.
Our main result is the following theorem. 
For any integers $a,b \in \mathbb{Z}$ with $a \le b$, let $[a,b]=\{\overline{a},\overline{a+1},\cdots,\overline{b}\} \subseteq \mathbb{Z}/r\mathbb{Z}$.


\begin{theorem}\label{th:impossible}
	Suppose $r$, $k$, and $A$ satisfy the following conditions.
\begin{enumerate}
	\item $r$ and $k$ are coprime and $1 \le k <\dfrac{r}{3}$.
	\item $\{0, k\} \subseteq A \subseteq [0, k] \subseteq \mathbb{Z}/r\mathbb{Z}$.
\end{enumerate}
Then, the following statements hold.
\begin{enumerate}
	\item If $(W, R)$ is Frame 1 and $\ini$ is $x$, then $(W, R)$ is an impossibility frame under $\Gamma:=\{\Box^{j} p, \lnot \Box^{j} p\}_{j \ge 1}$.
	\item If $(W, R)$ is Frame 2 satisfying $k$-symmetry and $\ini$ is $0$, then $(W, R)$ is an impossibility frame under $\Gamma:=\{(\Box \Diamond)^{j} \Box p, \lnot (\Box \Diamond)^{j} \Box p\}_{j \ge 0}$. 
\end{enumerate}
\end{theorem}
The proof proceeds by translating the modal agendas into a combinatorial structure on translates of $A$.
Using the semantic reduction from \Cref{sec:semantic_reduction}, we first rewrite the relevant modal formulas as propositions of the form $P_w$, where $P_w$ expresses that $A+w$ is contained in the valuation.
We then show that consistency of modal judgments is equivalent to a covering condition among these translates.
This allows us to identify a local overlap pattern that yields minimally inconsistent sets and one-step propagation.
Finally, using the coprimality of $r$ and $k$, we propagate this local structure around the cycle and obtain the strong path-connectedness required for impossibility.

Using \Cref{th:calcu reduc pro}, we rewrite the modal formulas in $\Gamma$ as propositions about translates of $A$.
Let $V$ denote the valuation of the Kripke model.
\begin{enumerate}
	\item Frame 1 : By the definition of Frame 1, we have $x : \Box \Phi \Leftrightarrow \land_{a \in A} a : \Phi$, and for all $w \in \mathbb{Z}/r\mathbb{Z}$, we have $w : \Box \Phi \Leftrightarrow w+\overline{k} : \Phi$. Therefore, for all $j \ge 1$, we have the following relation. We do not use Theorem \ref{th:calcu reduc pro} when $(W, R)$ is Frame 1. Also, we write $\overline{k}$ instead of abbreviating it as $k$ since $j$ is in $\mathbb{Z}$ and $\overline{k}$ is in $\mathbb{Z}/r\mathbb{Z}$.
	\begin{align*}
		x : \Box^{j} p &\Leftrightarrow \land_{a \in A} a : \Box^{j-1} p
		\Leftrightarrow \land_{a \in A} a+\overline{k} : \Box^{j-2} p \Leftrightarrow \cdots \Leftrightarrow \land_{a \in A} a+(j-1) \cdot \overline{k} : p \\
		&\Leftrightarrow \land_{a \in A} a+(j-1) \cdot \overline{k} \in V
		\Leftrightarrow A+(j-1) \cdot \overline{k} \subseteq V
	\end{align*}
		
	\item Frame 2 : From Theorem \ref{th:calcu reduc pro}, for all $j \in \mathbb{Z}_{\ge 0}$, we have the following relation.
	\begin{align*}
		0 : (\Box \Diamond)^{j} \Box p &\Leftrightarrow \overline{k} : (\Box \Diamond)^{(j-1)} \Box p
		\Leftrightarrow \cdots \Leftrightarrow j \cdot \overline{k} : \Box p
		\Leftrightarrow \land_{a \in A} j \cdot \overline{k}+a : p \\
		&\Leftrightarrow \land_{a \in A} j \cdot \overline{k}+a \in V
		\Leftrightarrow A+j \cdot \overline{k} \subseteq V
	\end{align*}
\end{enumerate}

	For all $w \in \mathbb{Z}/r\mathbb{Z}$, let the proposition $P_{w}$ mean $A+w \subseteq V$. From the above simplification, we regard $\Gamma$ as $\{P_{j \cdot \overline{k}}, \lnot P_{j \cdot \overline{k}}\}_{j \ge 0}$. Since $r$ and $k$ are coprime, for all $a \in \mathbb{Z}$, there exists some $b \in \mathbb{Z}$ and $c \in \mathbb{Z}_{\ge 0}$ such that $a=b \cdot r+c \cdot k$. Therefore, we have $a \equiv c \cdot k \mod r$, and thus $\{j \cdot \overline{k}\}_{j \ge 0}=\mathbb{Z}/r\mathbb{Z}$. Consequently, we regard $\Gamma$ as $\{P_{w}, \lnot P_{w}\}_{w \in \mathbb{Z}/r\mathbb{Z}}$. 
    This overlap-based mechanism plays a role analogous to shared propositional structure in the doctrinal paradox~\cite{Kornhauser1986, Pettit}.
    There, logical interconnections arise from overlap among propositional variables. Here, they arise from overlap among the translates $A+w$.

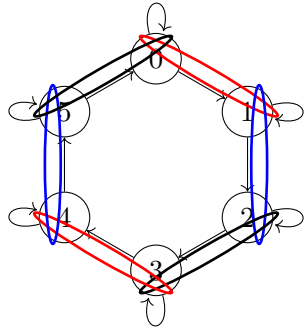
\begin{figure}[t]
\centering
\begin{tikzpicture}[ 
    every node/.style={circle, draw, minimum size=5mm}, scale=0.7,
    every edge/.style={draw, ->}
]

	\node (0) at (0,4) {$0$};
    	\node (1) at (1.73,3) {$1$};
	\node (2) at (1.73,1) {$2$};
	\node (3) at (0,0) {$3$};
	\node (4) at (-1.73,1) {$4$};
	\node (5) at (-1.73,3) {$5$};
	
    	\path (0) edge (1);
    	\path (1) edge (2);
    	\path (2) edge (3);
    	\path (3) edge (4);
    	\path (4) edge (5);
    	\path (5) edge (0);
    	\path (0) edge[loop above] (0);
    	\path (1) edge[loop right] (1);
    	\path (2) edge[loop right] (2);
    	\path (3) edge[loop below] (3);
    	\path (4) edge[loop left] (4);
    	\path (5) edge[loop left] (5);
	
	\draw[line width=1pt, red, rotate around={-30:(1,3.67)}] (1,3.67) ellipse [x radius=1.5cm, y radius=0.15cm];
	\draw[line width=1pt, blue, rotate around={90:(1.95, 2)}] (1.95, 2) ellipse [x radius=1.5cm, y radius=0.15cm];
	\draw[line width=1pt, rotate around={30:(1,0.33)}] (1,0.33) ellipse [x radius=1.5cm, y radius=0.15cm];
	
	\draw[line width=1pt, red, rotate around={-30:(-1,0.33)}] (-1,0.33) ellipse [x radius=1.5cm, y radius=0.15cm];
	\draw[line width=1pt, blue, rotate around={90:(-1.95, 2)}] (-1.95, 2) ellipse [x radius=1.5cm, y radius=0.15cm];
	\draw[line width=1pt, rotate around={30:(-1,3.67)}] (-1,3.67) ellipse [x radius=1.5cm, y radius=0.15cm];

\end{tikzpicture}
\caption{Frame 2 with $W=\mathbb{Z}/6\mathbb{Z}$, $k=1$, and $A=\{0,1\}$. The ellipses mean the sets of worlds $\{w,w+1\}$.}
\label{fig:logi chain}
\end{figure}


The following proposition characterizes the consistency of modal judgments in terms of a covering condition.
Its proof is deferred to Appendix \ref{sec:impossible agenda pro}.
\begin{prop}\label{lem:plus minus con}
	Let $J_{+}, J_{-} \subseteq \mathbb{Z}/r\mathbb{Z}$. Then $\Gamma_{0}:=\{P_{w}\}_{w \in J_{+}} \sqcup \{\lnot P_{w_{0}}\}_{w_{0} \in J_{-}}$ is consistent under $(W, R)$ if and only if for all $w_{0} \in J_{-}$, $A+w_{0} \not\subseteq \cup_{w \in J_{+}} A+w$.
\end{prop}


\begin{example}\label{ex:lem:plus minus con}
	Whether $(W, R)$ is Frame 1 or 2 makes no difference. Suppose $k=3$, $A=\{0,1,2,3\}$, and $J_{-}:=\{0\}$. $r$ is a natural integer larger than $9$ which is not divisible by $3$. 
    
    See Figure \ref{fig:pmcon}. Suppose $J_{+}=\{-3,3\}$. Then, we have $A+0 \setminus (\cup_{w \in J_{+}} A+w)=A+0 \setminus (A+(-3) \cup A+3)=\{1,2\} \neq \varnothing$. Set $V=A+(-3) \cup A+3=\{-3,-2,-1,0,3,4,5,6\}$, then $\lnot P_{0}$, $P_{-3}$, and $P_{3}$ are true. Therefore, $\Gamma_{0}=\{\lnot P_{0}, P_{-3}, P_{3}\}$ is consistent.

	 See Figure \ref{fig:pmincon}. Suppose $J_{+}=\{-1,2\}$. Then, we have $A+0 \subseteq A+(-1) \cup A+2=\cup_{w \in J_{+}} A+w$. If $P_{-1}$ and $P_{2}$ are true, we have $A+(-1) \subseteq V$ and $A+2 \subseteq V$. Thus, we have $A+0 \subseteq A+(-1) \cup A+2 \subseteq V$, and therefore $P_{0}$ is true. Consequently, $\Gamma_{0}=\{\lnot P_{0}, P_{-1}, P_{2}\}$ is inconsistent.

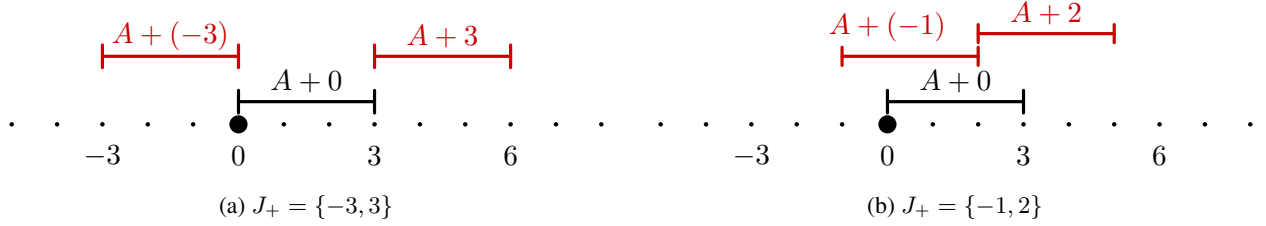
\begin{figure}[t]
\centering

\begin{subfigure}{0.48\textwidth}
\centering
\begin{tikzpicture}[x=6mm,y=6mm, line cap=round, line join=round, scale=1]

\tikzset{
  thickset/.style={line width=1.0pt, draw=black},
  redset/.style={line width=0.8pt, draw=red!70},
  a0fill/.style={draw=none, fill=gray!25},
  dotpt/.style={circle, fill=black, inner sep=0.6pt}
}

  \def\xL{-3}
  \def\xA{0}
  \def\xB{3}
  \def\xR{6}
  \def\xO{0}
  \def\xM{3}
  
  \draw[black, very thick] (\xO,0.5) -- (\xM,0.5);
  \draw[black, very thick] (\xO,0.75) -- (\xO,0.25);
  \draw[black, very thick] (\xM,0.75) -- (\xM,0.25);

  \draw[red!80!black, very thick] (\xL,1.5) -- (\xA,1.5);
  \draw[red!80!black, very thick] (\xB,1.5) -- (\xR,1.5);
  \draw[red!80!black, very thick] (\xL,1.75) -- (\xL,1.25);
  \draw[red!80!black, very thick] (\xA,1.75) -- (\xA,1.25);
  \draw[red!80!black, very thick] (\xB,1.75) -- (\xB,1.25);
  \draw[red!80!black, very thick] (\xR,1.75) -- (\xR,1.25);

  \node[red!80!black] at (-1.5,1.95) {$A+(-3)$};
  \node[red!80!black] at (4.5,1.95) {$A+3$};
  \node[black] at (1.5,0.95) {$A+0$};

  \node[black] at (\xL,-0.7) {$-3$};
  \node[black] at (\xA,-0.7) {$0$};
  \node[black] at (\xB,-0.7) {$3$};
  \node[black] at (\xR,-0.7) {$6$};

  \foreach \x in {-5,-4,...,8}
    \node[dotpt] at (\x,0) {};

  \fill (0,0) circle (0.2);

\end{tikzpicture}
\caption{$J_{+}=\{-3,3\}$}
\label{fig:pmcon}
\end{subfigure}
\hfill
\begin{subfigure}{0.48\textwidth}
\centering
\begin{tikzpicture}[x=6mm,y=6mm, line cap=round, line join=round, scale=1]

\tikzset{
  thickset/.style={line width=1.0pt, draw=black},
  redset/.style={line width=0.8pt, draw=red!70},
  a0fill/.style={draw=none, fill=gray!25},
  dotpt/.style={circle, fill=black, inner sep=0.6pt}
}

  \def\xL{-1}
  \def\xA{2}
  \def\xB{2}
  \def\xR{5}
  \def\xO{0}
  \def\xM{3}
  
  \draw[black, very thick] (\xO,0.5) -- (\xM,0.5);
  \draw[black, very thick] (\xO,0.75) -- (\xO,0.25);
  \draw[black, very thick] (\xM,0.75) -- (\xM,0.25);

  \draw[red!80!black, very thick] (\xL,1.5) -- (\xA,1.5);
  \draw[red!80!black, very thick] (\xB,2) -- (\xR,2);
  \draw[red!80!black, very thick] (\xL,1.7) -- (\xL,1.3);
  \draw[red!80!black, very thick] (\xA,1.7) -- (\xA,1.3);
  \draw[red!80!black, very thick] (\xB,2.2) -- (\xB,1.8);
  \draw[red!80!black, very thick] (\xR,2.2) -- (\xR,1.8);

  \node[red!80!black] at (0,2.15) {$A+(-1)$};
  \node[red!80!black] at (3.5,2.45) {$A+2$};
  \node[black] at (1.5,0.95) {$A+0$};

  \node[black] at (-3,-0.7) {$-3$};
  \node[black] at (0,-0.7) {$0$};
  \node[black] at (3,-0.7) {$3$};
  \node[black] at (6,-0.7) {$6$};

  \foreach \x in {-5,-4,...,8}
    \node[dotpt] at (\x,0) {};

  \fill (0,0) circle (0.2);

\end{tikzpicture}
\caption{$J_{+}=\{-1,2\}$}
\label{fig:pmincon}
\end{subfigure}

\caption{Examples of consistent and inconsistent choices of $J_{+}$}
\label{fig:pmcompare}
\end{figure}
\end{example}

The following definition isolates the set-theoretic structure that yields the required logical interconnections.

\begin{definition}\label{def:point min cov}
	Let $w_{0}, w_{1} \in \mathbb{Z}/r\mathbb{Z}$ and $S_{0} \subseteq \mathbb{Z}/r\mathbb{Z}$ with $|S_{0}| \ge 2$ and $w_{1} \in S_{0}$. The triple $(w_{0}, w_{1}, S_{0})$ is a pointed minimal cover if the following conditions are satisfied.
\begin{enumerate}
	\item Minimality: $\cup_{s \in S_{0}} A+s$ is a minimal cover of $A+w_{0}$. That is, for all $S_{1} \subseteq S_{0}$, $A+w_{0} \subseteq \cup_{s \in S_{1}} A+s$ if and only if $S_{1}=S_{0}$.
	\item Irreducibility: $A+w_{1} \not\subseteq \cup_{w \in (S_{0} \setminus \{w_{1}\}) \cup \{w_{0}\}} A+w$.
\end{enumerate}
\end{definition}


\begin{example}\label{ex:point min cov}
	See Figure \ref{fig: point min cov}. Note that whether $(W, R)$ is Frame 1 or 2 makes no difference. Suppose $k=3$ and $A=\{0,1,2,3\}$. $r$ is a natural integer larger than $9$ which is not divisible by $3$. Then $(0, k, \{-1,k\})$ is a pointed minimal cover. In fact, $A+(-1) \cup A+k$ is a minimal cover of $A+0=A$ (Minimality). Also, we have $2 \cdot k \in A+k \setminus (A+(-1) \cup A+0)=A+k \setminus (\cup_{s \in (\{-1, k\} \setminus \{k\}) \cup \{0\}} A+s)$ (Irreducibility).
	
\smallskip

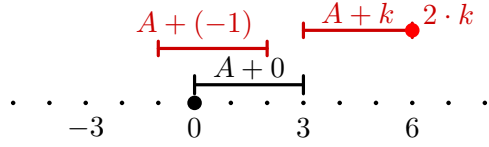
\begin{figure}[t]
\centering
\begin{tikzpicture}[x=8mm,y=8mm, line cap=round, line join=round,scale=0.6]

\tikzset{
  thickset/.style={line width=1.0pt, draw=black},
  redset/.style={line width=0.8pt, draw=red!70},
  a0fill/.style={draw=none, fill=gray!25},
  dotpt/.style={circle, fill=black, inner sep=0.6pt}
}

  \def\xL{-1}
  \def\xA{2}
  \def\xB{3}
  \def\xR{6}
  \def\xO{0}
  \def\xM{3}
  
  \draw[black, very thick] (\xO,0.5) -- (\xM,0.5);
  \draw[black, very thick] (\xO,0.75) -- (\xO,0.25);
  \draw[black, very thick] (\xM,0.75) -- (\xM,0.25);

  \draw[red!80!black, very thick]  (\xL,1.5) -- (\xA,1.5);
  \draw[red!80!black, very thick]  (\xB,2) -- (\xR,2);
  \draw[red!80!black, very thick] (\xL,1.7) -- (\xL,1.3);
  \draw[red!80!black, very thick] (\xA,1.7) -- (\xA,1.3);
  \draw[red!80!black, very thick] (\xB,2.2) -- (\xB, 1.8);
  \draw[red!80!black, very thick] (\xR,2.2) -- (\xR,1.8);

  \node[red!80!black] at (0,2.15) {$A+(-1)$};
  \node[red!80!black] at (4.5,2.45) {$A+k$};
  \node[black] at (1.5,0.95) {$A+0$};
  \node[red!80!black] at (7,2.45) {$2 \cdot k$};

  \node[black] at (-3,-0.7) {$-3$};
  \node[black] at (0,-0.7) {$0$};
  \node[black] at (3,-0.7) {$3$};
  \node[black] at (6,-0.7) {$6$};
  
\foreach \x in {-5,-4,...,8}
    \node[dotpt] at (\x,0) {};
    
\fill (0,0) circle (0.2);
\fill [red] (6,2) circle (0.2);

\end{tikzpicture}
\caption{A pointed minimal cover $(0, k, \{-1, k\})$}
\label{fig: point min cov}
\end{figure}
\end{example}

	The following lemma translates Definition \ref{def:point min cov} into logical interconnections for impossibility. We prove this lemma in Appendix  \ref{sec:impossible agenda pro}.

\begin{lemma}\label{lem:min cov good}
If $(w_{0}, w_{1},S_{0})$ is a pointed minimal cover, then the following statements hold.
\begin{enumerate}
	\item \label{lem:min cov good 1} $\{\lnot P_{w_{0}}\} \cup \{P_{s}\}_{s \in S_{0}}$ is minimally inconsistent.
	\item $P_{w_{1}} <_{0} P_{w_{0}}$.
\end{enumerate}
\end{lemma}

The following proposition constructs pointed minimal covers.
Its proof is deferred to Appendix \ref{sec:impossible agenda pro}.
\begin{prop}\label{prop:ex of minimal cover}
	For all $w \in \mathbb{Z}/r\mathbb{Z}$, there exists a subset $S_{0} \subseteq \mathbb{Z}/r\mathbb{Z}$ with $|S_{0}| \ge 2$ and $w+k \in S_{0}$ such that $(w, w+k, S_{0})$ is a pointed minimal cover.
\end{prop}

The proof of \Cref{th:impossible} relies on the partial overlap of the translates $A+w$.
\begin{proof}[Proof of Theorem \ref{th:impossible}]
	Recall the definition of an impossibility frame in Subsection \ref{subsec:JA}. In this proof, we write $\overline{k}$ instead of abbreviating it as $k$ since $j$ is in $\mathbb{Z}$ and $\overline{k}$ is in $\mathbb{Z}/r\mathbb{Z}$. From Proposition \ref{prop:ex of minimal cover}, for all $w \in \mathbb{Z}/r\mathbb{Z}$, choose $S_{0} \subseteq \mathbb{Z}/r\mathbb{Z}$ with $|S_{0}| \ge 2$ and $w+\overline{k} \in S_{0}$ such that $(w, w+\overline{k}, S_{0})$ is a pointed minimal cover. From Lemma \ref{lem:min cov good}, we have the following statements. Therefore, $(W, R)$ is an impossibility frame under $\Gamma$.
\begin{enumerate}
	\item For all $w \in \mathbb{Z}/r\mathbb{Z}$, $\{\lnot P_{w}\} \cup \{P_{s}\}_{s \in S_{0}}$ is minimally inconsistent and has more than two elements.
	\item For all $w \in \mathbb{Z}/r\mathbb{Z}$, we have $P_{w+\overline{k}} <_{0} P_{w}$. Since $r$ and $k$ are coprime, for all $w_{0}, w_{1} \in \mathbb{Z}/r\mathbb{Z}$, there exists some $j_{0} \in \mathbb{Z}_{\ge 1}$ such that $j_{0} \cdot \overline{k}=w_{1}-w_{0}$ in $\mathbb{Z}/r\mathbb{Z}$. Therefore, we have $P_{w_{1}}=P_{w_{0}+j_{0} \cdot \overline{k}} <_{0} P_{w_{0}+(j_{0}-1) \cdot \overline{k}} <_{0} \cdots <_{0} P_{w_{0}}$. Consequently, $\Gamma$ is strongly path-connected.
\end{enumerate}
\end{proof}

\section{Efficient non-dictatorial aggregation}\label{sec:altaggre}
The impossibility results of \Cref{sec:impossible agenda} show that proposition-wise independent aggregation fails even in a highly sparse modal setting.
This does not mean, however, that modal judgment aggregation is computationally or procedurally hopeless.
On the contrary, the semantic reduction of \Cref{sec:semantic_reduction} turns consistency checking into a small combinatorial covering problem, which makes constructive non-dictatorial aggregation possible.
The point of this section is not to circumvent the impossibility results of Section~4 under the same aggregation requirements.
Rather, we turn to constructive procedures that do not operate proposition-wise independently.
Our contribution in this section is to show that, in the present modal setting, the semantic reduction of Section~3.1 makes the required consistency checks efficient, and thereby yields efficiently implementable non-dictatorial aggregation procedures.

In this section, we exploit this reduction to obtain efficient aggregation procedures.
Subsection \ref{subsec:horn} presents a simple alternative based on forward chaining for Horn SAT \cite{Dowling1984}.
Subsection \ref{subsec:SAP} recalls the sequential majority procedure of \cite{Peleg}, and Subsections \ref{subsec:red} and \ref{subsec:red d1} show how Proposition \ref{lem:plus minus con} yields efficient implementations of its consistency-checking step.
The algorithm in Subsection \ref{subsec:red} applies to arbitrary $A$, whereas the algorithm in Subsection \ref{subsec:red d1} is restricted to $A=[0,k]$ but is more efficient.

Throughout this section, we assume the setting of \Cref{th:impossible}.
Rather than listing all individual judgment sets explicitly, we represent a profile by the counts $c(w):=\left|\{i\in N \mid P_w\in \Gamma_i=f(i)\}\right| (w\in \mathbb{Z}/r\mathbb{Z}),$
so the input consists of the frame parameters together with these acceptance counts.
We assume that $r$ and $k$ are given in binary, while $A\subseteq \mathbb{Z}/r\mathbb{Z}$ is given explicitly as a subset.
In the special case $A=[0,k]$, the set $A$ is determined by $r$ and $k$ and need not be given separately.

\subsection{Aggregation with forward chaining for Horn SAT}\label{subsec:horn}
We begin with a simple baseline procedure based on majority decisions, followed by Horn-style propagation.
Its purpose is to illustrate in a transparent way how the covering reduction can be used to enforce consistency efficiently.
The key point is that the covering relation from Proposition \ref{lem:plus minus con} induces Horn-type implications among the propositions $P_w$, so consistency can be enforced by forward propagation.
\begin{enumerate}
\item Prepare a list $\cover[w]\in\{0,1\}$ for all $w\in\mathbb{Z}/r\mathbb{Z}$, and initialize $\cover[w]=0$.
This list is used to represent a valuation $V$, where $\cover[w]=1$ and $\cover[w]=0$ mean $w\in V$ and $w\notin V$, respectively.

\item For each $w\in\mathbb{Z}/r\mathbb{Z}$, perform majority voting on $P_w$ and $\lnot P_w$ with anonymous tie-breaking.
In this paper, if $c(w)=n-c(w)$, then $P_w$ wins.
\begin{enumerate}
\item Suppose $P_w$ wins, that is, $c(w)\ge n-c(w)$.
Set $\cover[w_0]=1$ for all $w_0\in A+w$.
\item Suppose $\lnot P_w$ wins, that is, $c(w)<n-c(w)$.
Do nothing.
\end{enumerate}

\item Define the final valuation $V\subseteq W$ by $V=\{w\in\mathbb{Z}/r\mathbb{Z}\mid \cover[w]=1\}.$
The truth values of the propositions $P_w$ are then determined from this valuation.
\end{enumerate}
Step 1 takes $r$ steps.  Both Steps 2 and 3 take $r \cdot |A|$ steps. 
In total, the computation time is $O(r+r \cdot |A|+r \cdot |A|)=O(r \cdot |A|)$.

If $A+w \subseteq \cup_{w_{0} \in S} A+w_{0}$, then we have $(\lor_{w_{0} \in S} \lnot P_{w_{0}}) \lor P_{w}$. Thus, satisfying $\Gamma_{0}:=\{P_{w}\}_{w \in J_{+}} \cup \{\lnot P_{w_{0}}\}_{w_{0} \in J_{-}}$ corresponds to forward chaining for Horn SAT in \cite{Dowling1984}.

\subsection{The Sequential Majority Procedure}\label{subsec:SAP}
As a more principled non-dictatorial alternative, we consider the sequential majority procedure of Peleg and Zamir~\cite{Peleg}.
This procedure applies majority voting sequentially along a fixed exogenous order of the issues, represented here by a bijection
$\pi:\{0,1,\cdots,r-1\}\to\mathbb{Z}/r\mathbb{Z}$, where $\{0,1,\cdots,r-1\}\subseteq\mathbb{Z}$.
Our point is not to propose a new aggregation rule, but to show that, in our modal setting, its consistency-checking step admits efficient implementations.
Let $\Gamma$ be the agenda from \Cref{th:impossible}.

	The outputs are $J_{+}, J_{-} \subseteq \mathbb{Z}/r\mathbb{Z}$, which means that the group judgment is $\Gamma_{0}:=\{P_{w}\}_{w \in J_{+}} \cup \{\lnot P_{w_{0}}\}_{w_{0} \in J_{-}}$. First, set $J_{+}=J_{-}=\varnothing$.

\begin{enumerate}
	\item Choose between $P_{\pi(0)}$ and $\lnot P_{\pi(0)}$ by majority rule with anonymous tie-breaking and add $\pi(0)$ to $J_{+}$ or $J_{-}$.
	\item Suppose for all $0 \le q \le m-1$, either $\pi(q) \in J_{+}$ or $\pi(q) \in J_{-}$. 
	\begin{enumerate}
		\item \label{set a} If $\Gamma_{0} \sqcup \{\lnot P_{\pi(m)}\}$ is inconsistent, then add $\pi(m)$ to $J_{+}$.
		\item \label{set b} If $\Gamma_{0} \sqcup \{P_{\pi(m)}\}$ is inconsistent, then add $\pi(m)$ to $J_{-}$.
		\item \label{set c} Otherwise, choose between $P_{\pi(m)}$ and $\lnot P_{\pi(m)}$ by majority rule with anonymous tie-breaking and add $\pi(m)$ to $J_{+}$ or $J_{-}$.
	\end{enumerate}
\end{enumerate}

The final $J_{+}$ and $J_{-}$ are the outputs.

\medskip

	The bottleneck step in the above algorithm is consistency checking. To address this issue, Subsections \ref{subsec:red} and \ref{subsec:red d1} present two efficient algorithms by using Proposition \ref{lem:plus minus con}. The algorithm in Subsection \ref{subsec:red} covers arbitrary $A$. On the other hand, the algorithm in Subsection \ref{subsec:red d1} covers only $A=[0, k]$, but is more efficient than the one in Subsection \ref{subsec:red}.

\subsection{Sequential Majority Procedure in general}\label{subsec:red}

	In the sequential majority procedure, the set $\Gamma_{0}:=\{P_{w}\}_{w \in J_{+}} \cup \{\lnot P_{w_{0}}\}_{w_{0} \in J_{-}}$ is always consistent. In this subsection, we present an algorithm to determine whether $\Gamma_{0} \sqcup \{P_{\pi(m)}\}$ and $\Gamma_{0} \sqcup \{\lnot P_{\pi(m)}\}$ are consistent or not. The following two lemmata help reduce the computation time.

\begin{lemma}\label{lem:plus judge}
If $\Gamma_{0}:=\{P_{w}\}_{w \in J_{+}} \cup \{\lnot P_{w_{0}}\}_{w_{0} \in J_{-}}$ is consistent, then the following conditions are equivalent.
\begin{enumerate}
	\item \label{lem:plus judge 1} $\Gamma_{0} \sqcup \{\lnot P_{\pi(m)}\}$ is inconsistent.
	\item \label{lem:plus judge 2} $A+\pi(m) \subseteq \cup_{w \in J_{+}} A+w$.
\end{enumerate}
\end{lemma}

\begin{proof}
	From Proposition \ref{lem:plus minus con}, we have $(\ref{lem:plus judge 2}) \rightarrow (\ref{lem:plus judge 1})$. We prove the converse. Suppose $(\ref{lem:plus judge 1})$. From Proposition \ref{lem:plus minus con}, there exists some $w_{1} \in J_{-} \cup \{\pi(m)\}$ such that $A+w_{1} \subseteq \cup_{w \in J_{+}} A+w$. We prove $w_{1}=\pi(m)$. Assume $w_{1} \in J_{-}$. From Proposition \ref{lem:plus minus con},  $\Gamma_{0}$ is inconsistent, which contradicts the consistency of $\Gamma_{0}$. Therefore, we have $w_{1}=\pi(m)$ and finish the proof.
\end{proof}

	Intuitively, Lemma \ref{lem:minus judge} means that we have to work only near $A+\pi(m)$.

\begin{lemma}\label{lem:minus judge}
If $\Gamma_{0}:=\{P_{w}\}_{w \in J_{+}} \cup \{\lnot P_{w_{0}}\}_{w_{0} \in J_{-}}$ is consistent, then the following conditions are equivalent.
\begin{enumerate}
	\item \label{lem:minus judge 1} $\Gamma_{0} \sqcup \{P_{\pi(m)}\}$ is inconsistent.
	\item \label{lem:minus judge 2} There exists some $w_{0} \in J_{-}$ with $A+w_{0} \cap A+\pi(m) \neq \varnothing$ such that $A+w_{0} \subseteq \cup_{w \in J_{+} \cup \{\pi(m)\}} A+w$.
\end{enumerate}
\end{lemma}

\begin{proof}
	From Proposition \ref{lem:plus minus con}, we have $(\ref{lem:minus judge 2}) \rightarrow (\ref{lem:minus judge 1})$. We prove the converse. Suppose $(\ref{lem:minus judge 1})$. From Proposition \ref{lem:plus minus con}, there exists some $w_{0} \in J_{-}$ such that $A+w_{0} \subseteq \cup_{w \in J_{+} \cup \{\pi(m)\}} A+w$. We prove $A+w_{0} \cap A+\pi(m) \neq \varnothing$. Assume $A+w_{0} \cap A+\pi(m)=\varnothing$. Then, we have $A+w_{0} \subseteq \cup_{w \in J_{+}} A+w$. From Proposition \ref{lem:plus minus con}, $\Gamma_{0}$ is inconsistent, which contradicts the consistency of $\Gamma_{0}$. Therefore, we have $A+w_{0} \cap A+\pi(m) \neq \varnothing$.
\end{proof}

	We estimate the computation time in step (\ref{set a}), (\ref{set b}) and (\ref{set c}). Prepare the list $\cover[w] \in \{0,1\}$, where $w \in \mathbb{Z}/r\mathbb{Z}$. First, we set $\cover[w]=0$ for all $w \in \mathbb{Z}/r\mathbb{Z}$. 
    After the procedure, for all $w \in \mathbb{Z}/r\mathbb{Z}$, $\cover[w]=1$ means $w \in \cup_{w_{0} \in J_{+}} A+w_{0}$.

	In step (\ref{set a}), from Lemma \ref{lem:plus judge}, it is enough to check if $A+\pi(m) \subseteq \cup_{w_{0} \in J_{+}} A+w_{0}$. By using the list $\cover[w]$, step (\ref{set a}) takes $|A+\pi(m)|=|A|$ steps. If $A+\pi(m) \subseteq \cup_{w_{0} \in J_{+}} A+w_{0}$, then $\Gamma_{0} \sqcup \{\lnot P_{\pi(m)}\}$ is inconsistent. Therefore, we add $\pi(m)$ to $J_{+}$. From $A+\pi(m) \subseteq \cup_{w_{0} \in J_{+}} A+w_{0}$, we need not update the list $\cover[w]$.
	
	In step (\ref{set b}), from Lemma \ref{lem:minus judge}, it is enough to check if there exists some $w_{0} \in J_{-}$ with $A+w_{0} \cap A+\pi(m) \neq \varnothing$ such that $A+w_{0} \subseteq \cup_{w \in J_{+} \cup \{\pi(m)\}} A+w$. We have the following equivalence transformation.
	\begin{align*}
		& A+w_{0} \cap A+\pi(m) \neq \varnothing \quad \Leftrightarrow \exists a_{0}, a_{1} \in A, a_{0}+w_{0}=a_{1}+\pi(m) \\
		\Leftrightarrow \quad & \exists a_{0}, a_{1} \in A, w_{0}=a_{1}-a_{0}+\pi(m) \quad \Rightarrow w_{0} \in [-k, k]+\pi(m) \quad (A \subseteq [0,k])
	\end{align*}
	We consider two cases.
	\begin{enumerate}
		\item If $|A| \le (2 \cdot k)^{1/2}$. From $\exists a_{0}, a_{1} \in A, w_{0}=a_{1}-a_{0}+\pi(m)$, the number of candidates of $w_{0} \in J_{-}$ is $|A|^2$. Therefore, step (\ref{set b}) takes $|A|^2 \cdot |A+w_{0}|=|A|^3$ steps.
		\item If $|A| \ge (2 \cdot k)^{1/2}$. From $w_{0} \in [-k, k]+\pi(m)$, the number of candidates of $w_{0} \in J_{-}$ is $|[-k, k]| \approx 2 \cdot k$. Therefore, step (\ref{set b}) takes $(2 \cdot k) \cdot |A+w_{0}|=2 \cdot k \cdot |A|$ steps.
	\end{enumerate}
	
	In step (\ref{set c}), majority rule takes one step. If we add $\pi(m)$ to $J_{+}$, we substitute $1$ into $\cover[w]$ for all $w \in A+\pi(m)$. If we add $\pi(m)$ to $J_{-}$, then do nothing. Therefore, step (\ref{set c}) takes $1+|A+\pi(m)| \approx |A|$ steps.
	
	Since we execute these processes $r$ times, the computation time is $O(r \cdot \min \{|A|^3,k \cdot |A|\})$ as follows.

\subsection{Sequential Majority Procedure in the case $A=[0,k]$}\label{subsec:red d1}

	In this subsection, we present an algorithm for $A=[0,k]$, which is more efficient than the one in Subsection \ref{subsec:red}.
    
    Prepare the list $\cvle[w], \cvri[w]$ for all $w \in \mathbb{Z}/r\mathbb{Z}$. Define
	\begin{align*}
		\cvle[w]:=\min \{t \in [0, k] \mid w-t \in J_{+}\}, \quad
		\cvri[w]:=\min \{t \in [0, k] \mid w+t \in J_{+}\}. \\
	\end{align*}
	If the corresponding sets are empty, define $\cvle[w]=k+1$ or $\cvri[w]=k+1$.
	
	Intuitively, $w-\cvle[w] \in J_{+}$ (resp. $w+\cvri[w] \in J_{+}$) is the closest index to $w$ within distance $k$ on the left (resp. right). Equivalently, among the intervals $\{A+s\}_{s \in J_{+}}$ which intersect $A+w$ from the left (resp. right), $A+(w-\cvle[w])$ is the rightmost one and $A+(w+\cvri[w])$ is the leftmost one. 
	
	The following lemma uses these two lists to convert the set cover problem into a simple numerical condition. In Figure \ref{fig:algo two inter}, $w_{0}$ and $w_{1}$ equal $w-\cvle[w]$ and $w+\cvri[w]$, respectively.
	
\begin{lemma}\label{lem:check one}
	For all $w \in \mathbb{Z}/r\mathbb{Z}$, $A+w \subseteq \cup_{w_{0} \in J_{+}} A+w_{0}$ if and only if $\cvle[w]+\cvri[w] \le k+1$. 
\end{lemma}

\begin{figure}[t]
\centering
\begin{tikzpicture}[x=8mm,y=8mm, line cap=round, line join=round,scale=0.6]
  \def\xL{-3}   
  \def\xR{10.5}    
  \def\xO{0.0}    
  \def\xM{6.0}    
  \def\xA{3}    
  \def\xB{4.5}    

  \draw[black, very thick] (\xO,0) -- (\xM,0);
  \draw[black, very thick] (\xO,0.22) -- (\xO,-0.22);
  \draw[black, very thick] (\xM,0.22) -- (\xM,-0.22);
  \node[below,font=\tiny] at (\xO,-0.25) {$w$};
  \node[below,font=\tiny] at (\xM,-0.25) {$w+k$};
  \node[,font=\tiny] at ({(\xO+\xM)/2},-0.45) {$A+w$};

  \draw[red!80!black, very thick]  (\xL,1.2) -- (\xA,1.2);
  \draw[red!80!black, very thick]  (\xB,1.2) -- (\xR,1.2);
  \draw[red!80!black, very thick] (\xL,1.42) -- (\xL,0.98);
  \draw[red!80!black, very thick] (\xA,1.42) -- (\xA,0.98);
  \draw[red!80!black, very thick] (\xB,1.42) -- (\xB,0.98);
  \draw[red!80!black, very thick] (\xR,1.42) -- (\xR,0.98);

  \node[red!80!black,font=\tiny] at ({(\xL+\xL+\xA)/3},1.55) {$A+w_0$};
  \node[red!80!black,font=\tiny] at ({(\xB+\xR)/2},1.55) {$A+w_1$};

  \node[red!80!black,font=\tiny] at (\xL,0.1) {$w_{0}$};
  \node[red!80!black,font=\tiny] at (\xA,1.7) {$w_{0}+k$};
  \node[red!80!black,font=\tiny] at (\xB,0.5) {$w_{1}$};
  \node[red!80!black,font=\tiny] at (\xR,0.1) {$w_{1}+k$};
\end{tikzpicture}
\caption{A covering of $A+w$}
\label{fig:algo two inter}
\end{figure}
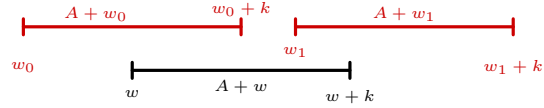

	We estimate the computation time in step (\ref{set a}), (\ref{set b}) and (\ref{set c}). First, we set $\cvle[w]=\cvri[w]=k+1$ for all $w \in \mathbb{Z}/r\mathbb{Z}$. 

	In step (\ref{set a}), from Lemma \ref{lem:plus judge}, it is enough to check if $A+\pi(m) \subseteq \cup_{w \in J_{+}} A+w$. From Lemma \ref{lem:check one}, it is enough to check if $\cvle[\pi(m)]+\cvri[\pi(m)] \le k+1$. This checking takes one step. If we add $\pi(m)$ to $J_{+}$, we must update the list $\cvle[w]$ and $\cvri[w]$. It is enough to do so for $w \in \mathbb{Z}/r\mathbb{Z}$ such that $A+w \cap A+\pi(m) \neq \varnothing$. $A+w \cap A+\pi(m) \neq \varnothing$ is equivalent to $w \in \pi(m)+[-k, k]$. Therefore, step (\ref{set a}) takes $2 \cdot k$ steps.
	
	In step (\ref{set b}), from Lemma \ref{lem:minus judge}, it is enough to check if there exists some $w_{0} \in J_{-}$ with $A+w_{0} \cap A+\pi(m) \neq \varnothing$ such that $A+w_{0} \subseteq \cup_{w \in J_{+} \cup \{\pi(m)\}} A+w$. By Lemma \ref{lem:check one}, this condition can be rewritten as follows. Note that from the definition of $J_{-}$ and $\pi(m)$, we have $w_{0} \neq \pi(m)$.
\begin{enumerate}
	\item See Figure \ref{fig:algo-two-inter-cases-minus}. Suppose $w_{0}-\pi(m) \in [-k, -1]$. If we add $\pi(m)$ to $J_{+}$, then $\cvle[w_{0}]$ and $\cvri[w_{0}]$ change to $\cvle[w_{0}]$ and $\min(\cvri[w_{0}], \pi(m)-w_{0})$, respectively. Therefore, it is enough to check if $\cvle[w_{0}]+\min(\cvri[w_{0}], \pi(m)-w_{0}) \le k+1$.
\item See Figure \ref{fig:algo-two-inter-cases-plus}. Suppose $w_{0}-\pi(m) \in [1, k]$. If we add $\pi(m)$ to $J_{+}$, then $\cvle[w_{0}]$ and $\cvri[w_{0}]$ change to $\min(\cvle[w_{0}], w_{0}-\pi(m))$ and $\cvri[w_{0}]$, respectively. Therefore, it is enough to check $\min(\cvle[w_{0}], w_{0}-\pi(m))+\cvri[w_{0}] \le k+1$.
\end{enumerate}
Therefore, step (\ref{set b}) takes $2 \cdot k$ steps.

\begin{figure}[t]
  \centering

  \begin{subfigure}[t]{0.48\textwidth}
    \centering
    \begin{tikzpicture}[x=8mm,y=8mm, line cap=round, line join=round,scale=0.6]
      \def\xL{-3}
      \def\xR{10.5}
      \def\xO{0}
      \def\xM{6}
      \def\xA{3}
      \def\xB{4.5}

      \draw[red!80!black, very thick] (\xO,0) -- (\xM,0);
      \draw[red!80!black, very thick] (\xO,0.22) -- (\xO,-0.22);
      \draw[red!80!black, very thick] (\xM,0.22) -- (\xM,-0.22);

      \node[red!80!black, below,font=\tiny] at (\xO,-0.25) {$\pi(m)$};
      \node[red!80!black, below,font=\tiny] at (\xM,-0.25) {$\pi(m)+k$};
      \node[red!80!black, font=\tiny] at ({(\xO+\xM)/2},-0.45) {$A+\pi(m)$};

      \draw[black, very thick] (\xL,1.2) -- (\xA,1.2);
      \draw[black, very thick] (\xL,1.42) -- (\xL,0.98);
      \draw[black, very thick] (\xA,1.42) -- (\xA,0.98);

      \node[black,font=\tiny] at ({(\xL+\xA)/2},1.55) {$A+w_0$};
      \node[black,font=\tiny] at (\xL,0.5) {$w_0$};
      \node[black,font=\tiny] at (\xA,1.7) {$w_0+k$};
    \end{tikzpicture}
    \caption{$w_{0}-\pi(m)\in[-k,-1]$}
    \label{fig:algo-two-inter-cases-minus}
  \end{subfigure}
  \hfill
  \begin{subfigure}[t]{0.48\textwidth}
    \centering
    \begin{tikzpicture}[x=8mm,y=8mm, line cap=round, line join=round,scale=0.6]
      \def\xL{-3}
      \def\xR{10.5}
      \def\xO{0}
      \def\xM{6}
      \def\xA{3}
      \def\xB{4.5}

      \draw[red!80!black, very thick] (\xO,0) -- (\xM,0);
      \draw[red!80!black, very thick] (\xO,0.22) -- (\xO,-0.22);
      \draw[red!80!black, very thick] (\xM,0.22) -- (\xM,-0.22);

      \node[red!80!black, below,font=\tiny] at (\xO,-0.25) {$\pi(m)$};
      \node[red!80!black, below,font=\tiny] at (\xM,-0.25) {$\pi(m)+k$};
      \node[red!80!black, font=\tiny] at ({(\xO+\xM)/2},-0.45) {$A+\pi(m)$};

      \draw[black, very thick] (\xB,1.2) -- (\xR,1.2);
      \draw[black, very thick] (\xB,1.42) -- (\xB,0.98);
      \draw[black, very thick] (\xR,1.42) -- (\xR,0.98);

      \node[black,font=\tiny] at ({(\xB+\xR)/2},1.55) {$A+w_0$};
      \node[black,font=\tiny] at (\xB,0.5) {$w_0$};
      \node[black,font=\tiny] at (\xR,0.5) {$w_0+k$};
    \end{tikzpicture}
    \caption{$w_{0}-\pi(m)\in[1,k]$}
    \label{fig:algo-two-inter-cases-plus}
  \end{subfigure}

  \caption{Two cases of the relative position of $A+w_0$ and $A+\pi(m)$.}
  \label{fig:algo-two-inter-cases}
\end{figure}
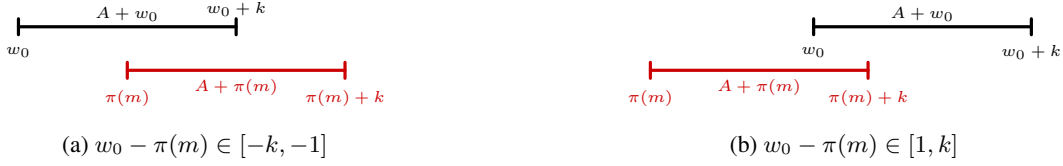
	
	In step (\ref{set c}), majority rule takes one step. If we add $\pi(m)$ to $J_{+}$, we must update the list $\cvle[w]$ and $\cvri[w]$. By the same argument as in step (\ref{set a}), step (\ref{set c}) takes $2 \cdot k$ steps. If we add $\pi(m)$ to $J_{-}$, then do nothing.
	
	Since we execute these processes $r$ times, the computation time is $O(r \cdot (k+2 \cdot k+2 \cdot k))=O(r \cdot k)$.

\section{Concluding Remarks}\label{sec:conclusion}
This paper shows that Arrow-type impossibility does not disappear when judgment aggregation is restricted to genuinely modal judgments.  
Even after setting aside fact-based aggregation, very sparse modal patterns already generate enough logical interconnection to force dictatorship.  
At the technical level, our results identify a concrete route from frame geometry to impossibility: semantic reduction turns modal formulas into a finite combinatorial form, and local overlap structure then yields the global connectivity required for dictatorship.  
The same reduction also provides the algorithmic leverage needed for efficient implementations of non-dictatorial aggregation procedures.

More broadly, our results clarify the scope of Arrow-type impossibility in computational social choice.  
They show that this phenomenon is not tied only to propositional agendas built from bare factual judgments.  
Even in a setting where the objects of aggregation are genuinely modal judgments, semantic structure alone can generate the interconnections needed for dictatorship.  
This suggests that modal semantics should be viewed not merely as a richer language for aggregation, but as a direct source of aggregation-theoretic obstruction.



\appendix

\begin{appendices}
The appendices are organized as follows.
Appendix \ref{sec:related} provides additional background and related work, while Appendices B--E collect definitions and proofs deferred from the main text.

\section{Additional background and related work}\label{sec:related}
\paragraph{Classical impossibility and abstraction.}
Logical interconnections have long been a central source of impossibility in social choice and judgment aggregation.
As shown in Table~\ref{tab:condorcet}, Condorcet's paradox \cite{Con} shows that issue-by-issue majority voting can turn individually rational preferences into a cyclic social preference.
Arrow \cite{Arrow} then shifted attention from a single paradoxical rule to an axiomatic analysis of all aggregation rules, proving that any independent preference aggregation is impossible unless it is dictatorial.
A parallel line of work begins with the doctrinal paradox \cite{Kornhauser1986} or the discursive dilemma \cite{Pettit}.
As shown in Table~\ref{tab:doctrinal}, majority voting on $p$, $q$, and $p \land q$ may turn consistent individual judgments into an inconsistent collective judgment.
List and Pettit \cite{Listdicur} generalized this paradox into an impossibility theorem and thereby founded judgment aggregation in propositional logic.

The focus then moved from particular paradoxes to more abstract frameworks.
Dietrich \cite{Dietrich} proposed a general framework for judgment aggregation across logics.
However, since Arrow's impossibility theorem does not overlap with the impossibility result of List and Pettit \cite{Listdicur}, it remained open how preference aggregation relates to judgment aggregation.
To bridge this gap, Dietrich and List \cite{List} showed, within Dietrich's framework, that any independent aggregation is dictatorial in a broad class of logics, including Arrow's framework.
Dokow and Holzman \cite{Dokow2010} further abstracted this line of work by formulating binary aggregation, where one studies the set of feasible truth-value vectors over an agenda.
They identified necessary and sufficient conditions on agendas under which Independence entails dictatorship.
Unlike \cite{List}, Dokow and Holzman \cite{Dokow2010} treat interconnections directly, without reference to the underlying logic.
Indeed, the frameworks of \cite{List} and \cite{Dokow2010} already include modal logic.
However, they take the logical interconnections of agendas as given, rather than asking how such interconnections are generated by the logic itself.
\begin{table}[hbt]
\centering
\scalebox{1}{
\begin{subtable}{0.55\linewidth}
\centering
\[
\begin{array}{ccccc}
                & a>b & b>c & c>a & \text{Preferences} \\
\text{Ind}1 & \text{Yes} & \text{Yes} & \text{No} & a>b>c \\
\text{Ind}2 & \text{No} & \text{Yes} & \text{Yes} & b>c>a \\
\text{Ind}3 & \text{Yes} & \text{No} & \text{Yes} & c>a>b \\
\text{MV}   & \text{Yes} & \text{Yes} & \text{Yes} & a>b>c>a
\end{array}
\]
\caption{Condorcet's paradox}
\label{tab:condorcet}
\end{subtable}
\hspace{3em}
\begin{subtable}{0.40\linewidth}
\centering
\[
\begin{array}{cccc}
                & p & q & p \land q \\
\text{Judge}1 & \text{Yes} & \text{Yes} & \text{Yes} \\
\text{Judge}2 & \text{Yes} & \text{No} & \text{No} \\
\text{Judge}3 & \text{No} & \text{Yes} & \text{No} \\
\text{MV}   & \text{Yes} & \text{Yes} & \text{No}
\end{array}
\]
\caption{The doctrinal paradox}
\label{tab:doctrinal}
\end{subtable}
}
\caption{Majority voting paradoxes}
\label{tab:paradoxes}
\end{table}
\paragraph{Logic-specific studies.}
Logic-specific properties have also been studied in judgment aggregation.
Dietrich \cite{Dietrich} provided a general framework across logics, and subsequent work within this framework has clarified how proof systems and consistency notions affect aggregation.
For example, Porello \cite{Porello} showed that, under some non-classical logics, majority voting yields consistent judgments, whereas Wen \cite{Wen} showed that Dietrich--List type impossibility persists even in non-monotonic logics.
These studies show that the choice of logic matters for judgment aggregation.
Our emphasis is different.
In contrast to Porello \cite{Porello} and Wen \cite{Wen}, we show not merely that changing the underlying logic changes consistency or impossibility, but that a semantic structure itself determines the logical interconnections that entail dictatorship.
This is specific to our setting, where frame geometry rather than syntactic richness drives the result.

\paragraph{Modal logic in judgment aggregation.}
Modal logic has also appeared in the judgment-aggregation literature, but mainly as a meta-language for describing aggregation.
By contrast, our paper treats modal propositions themselves as the objects of aggregation.
Novaro et al.\ \cite{Novaro} explain that such formalizations aim both to verify known results and to discover new ones by means of automated reasoning, and they identify Pauly \cite{Pauly} and \AA gotnes et al.\ \cite{Agotnes} as notable examples.
\AA gotnes et al.\ \cite{Agotnes} point out that Pauly \cite{Pauly} axiomatizes the result of judgment aggregation, whereas their framework internalizes profiles and aggregation rules.
As noted by Novaro et al.\ \cite{Novaro}, however, these approaches rely on newly designed logical languages and therefore do not directly support the use of existing automated-reasoning tools.
To address this issue, Novaro et al.\ \cite{Novaro} encoded judgment aggregation in Dynamic Logic of Propositional Assignments, originally introduced by Balbiani et al.\ \cite{Balbiani}.

\section{Definition of $\mathbb{Z}/r\mathbb{Z}$ and the addition operation $+$ on it}\label{appendix: addition}

In this section, we define $\mathbb{Z}/r\mathbb{Z}$ and the addition operation $+$ on it.

\begin{definition}
Let $r \in \mathbb{Z}_{\ge 2}$. For all $a, b \in \mathbb{Z}$, we set $a \equiv b$ if $a-b$ is divisible by $r$. That is, there exists some $m \in \mathbb{Z}$ such that $a-b=r \cdot m$.
\end{definition}

\begin{example}
	Suppose $r=3$. For all $m \in \mathbb{Z}$, we have $3 \cdot m \equiv 0$, $3 \cdot m+1 \equiv 1$, and $3 \cdot m +2 \equiv 2$.
\end{example}

\begin{lemma}\label{lem:equiv}
	The relation $\equiv$ is an equivalence relation on $\mathbb{Z}$. That is, the following assertions hold.
\begin{enumerate}
	\item \label{lem:equiv 1} $a \equiv a$.
	\item \label{lem:equiv 2} $a \equiv b$ implies $b \equiv a$.
	\item \label{lem:equiv 3} $a \equiv b$ and $b \equiv c$ imply $a \equiv c$.
\end{enumerate}
\end{lemma}

\begin{proof}
\begin{enumerate}
	\item From $a-a=0=r \cdot 0$, we have $a \equiv a$.
	\item Suppose $a \equiv b$. Then there exists some $m \in \mathbb{Z}$ such that $a-b=r \cdot m$. From $b-a=r \cdot (-m)$, we have $b \equiv a$.
	\item Suppose $a \equiv b$ and $b \equiv c$. Then we have the following two assertions.
	\begin{enumerate}
		\item There exists some $m_{1} \in \mathbb{Z}$ such that $a-b=r \cdot m_{1}$. 
		\item There exists some $m_{2} \in \mathbb{Z}$ such that $b-c=r \cdot m_{2}$. 
	\end{enumerate}
	From $a-c=(a-b)+(b-c)=r \cdot (m_{1}+m_{2})$, we have $a \equiv c$.
\end{enumerate}
\end{proof}

\begin{definition}\label{def:set and cano}
We define the set $\mathbb{Z}/r\mathbb{Z}$ as follows.
\begin{align*}
	\overline{a}:=& \quad \{c \in \mathbb{Z} \mid c \equiv a\} \quad (a \in \mathbb{Z}) \\
	\mathbb{Z}/r\mathbb{Z}:=& \quad \{\overline{a} \mid a \in \mathbb{Z}\} \\
\end{align*}
\end{definition}

Since all integers are of the form $r \cdot m+i$ with $i=0, 1, \cdots , r-1$, we have the following equation.
\begin{align*}
	\mathbb{Z}/r\mathbb{Z}=\{\overline{0}, \overline{1}, \cdots, \overline{r-1} \}
\end{align*}
Thus, two integers represent the same element of $\mathbb{Z}/r\mathbb{Z}$
exactly when they have the same remainder modulo $r$. That is, $\overline{a}=\overline{b}$ if and only if $a \equiv b$.

\begin{definition}
For all $a, b \in \mathbb{Z}$, we define the addition operator $+$ on $\mathbb{Z}/r\mathbb{Z}$ as follows.
\begin{align*}
	\overline{a}+\overline{b}:=\overline{a+b}
\end{align*}
\end{definition}

\begin{lemma}
	The addition operation $+$ on $\mathbb{Z}/r\mathbb{Z}$ is well-defined. That is, for all $a_{1}, a_{2}, b_{1}, b_{2} \in \mathbb{Z}$, $\overline{a_{1}}=\overline{a_{2}}$ and $\overline{b_{1}}=\overline{b_{2}}$ imply $\overline{a_{1}+b_{1}}=\overline{a_{2}+b_{2}}$.
\end{lemma}

\begin{proof}
Suppose $\overline{a_{1}}=\overline{a_{2}}$ and $\overline{b_{1}}=\overline{b_{2}}$. Then we have $a_{1} \equiv a_{2}$ and $b_{1} \equiv b_{2}$. Thus, we have the following assertions.
\begin{enumerate}
	\item There exists some $m \in \mathbb{Z}$ such that $a_{1}-a_{2}=r \cdot m$. 
	\item There exists some $n \in \mathbb{Z}$ such that $b_{1}-b_{2}=r \cdot n$. 
\end{enumerate}
	From $(a_{1}+b_{1})-(a_{2}+b_{2})=(a_{1}-a_{2})+(b_{1}-b_{2})=r \cdot (m+n)$, we have $a_{1}+b_{1} \equiv a_{2}+b_{2}$, and thus $\overline{a_{1}+b_{1}}=\overline{a_{2}+b_{2}}$.
\end{proof}

\begin{definition}
For all $a \in \mathbb{Z}_{\ge 0}$ and $b \in \mathbb{Z}$, we define $a \cdot \overline{b}$ as follows.
\begin{align*}
	a \cdot \overline{b}:=\underset{\text{$a$ times}}{\overline{b}+\overline{b}+\cdots+\overline{b}}
\end{align*}
Trivially, we have $a \cdot \overline{b}=\overline{a \cdot b}$.
\end{definition}

\section{Proof of Theorem \ref{th:calcu reduc pro}}\label{sec:reduc modal}

	In this section, we fix a Frame 2 $(\mathbb{Z}/r\mathbb{Z}, R)$ and prove Theorem \ref{th:calcu reduc pro}, which we introduce in Subsection \ref{sec:semantic_reduction}. We restate Frame 2, $k$-symmetry, and Theorem \ref{th:calcu reduc pro}.

Frame 2 is as follows.
\begin{align*}
W:=\mathbb{Z}/r\mathbb{Z}, \quad R:=\{(w, w+a) \mid w \in \mathbb{Z}/r\mathbb{Z}, a \in A\}
\end{align*}

\begin{definition}
	The set $A$ is $k$-symmetric if for all $a \in A$, $k-a \in A$.
\end{definition}

\begin{theorem}\label{th:calcu reduc pro ag}
	Suppose $(W, R)$ is Frame 2. If the set $A$ is $k$-symmetric, then for all worlds $w \in \mathbb{Z}/r\mathbb{Z}$, modal formulas $\Phi$, and valuations $V$, $M=(\mathbb{Z}/r\mathbb{Z}, R, V) \vDash w : \Box \Diamond \Box \Phi$ if and only if $M=(\mathbb{Z}/r\mathbb{Z}, R, V) \vDash w+k : \Box \Phi$.
\end{theorem}

	Before proving Theorem \ref{th:calcu reduc pro ag}, we show Lemma \ref{lem:leq}.

\begin{lemma}\label{lem:leq}
	Suppose $(W, R)$ is Frame 2. If the set $A$ is $k$-symmetric, then for all worlds $w \in \mathbb{Z}/r\mathbb{Z}$, modal formulas $\Phi$, and valuations $V \subseteq \mathbb{Z}/r\mathbb{Z}$, $(\mathbb{Z}/r\mathbb{Z}, R, V) \vDash w+k : \Phi$ implies $(\mathbb{Z}/r\mathbb{Z}, R, V) \vDash w : \Box \Diamond \Phi$ and $(\mathbb{Z}/r\mathbb{Z}, R, V) \vDash w : \Diamond \Box \Phi$ implies $(\mathbb{Z}/r\mathbb{Z}, R, V) \vDash w+k : \Phi$.
\end{lemma}

Example \ref{ex:lemeq} illustrates that $(\mathbb{Z}/r\mathbb{Z}, R, V) \vDash w+k : \Phi$ implies $(\mathbb{Z}/r\mathbb{Z}, R, V) \vDash w : \Box \Diamond \Phi$ in Lemma \ref{lem:leq}.

\begin{example}\label{ex:lemeq}
	Suppose $W=\mathbb{Z}/6\mathbb{Z}$, $k=2$, $A=\{0,1,2\}$, and $w=0$. Then $A$ is $k$-symmetric. We expand $0 : \Box \Diamond \Phi$ as follows.
	\begin{equation}
	\begin{aligned}
		0 : \Box \Diamond \Phi 
		\Leftrightarrow \land_{w_{2} \in \mathbb{Z}/r\mathbb{Z}, (0, w_{2}) \in R} w_{2} : \Diamond \Phi
		\Leftrightarrow \land_{w_{2} \in 0+A=\{0,1,2\}} w_{2} : \Diamond \Phi
		\Leftrightarrow 0 : \Diamond \Phi \land 1 : \Diamond \Phi \land 2 : \Diamond \Phi \label{eq:k0sym 1}
	\end{aligned}
	\end{equation}
	We further expand $0 : \Diamond \Phi$, $1 : \Diamond \Phi$, and $2 :\Diamond \Phi$ as follows.
	\begin{equation}
	\begin{aligned}
		& 0 : \Diamond \Phi \Leftrightarrow \lor_{w_{2} \in 0+A} w_{2} : \Phi
		\Leftrightarrow 0:\Phi \lor 1:\Phi \lor \textcolor{red}{2} : \Phi \\
		& 1:\Diamond \Phi \Leftrightarrow \lor_{w_{2} \in 1+A} w_{2} : \Phi
		\Leftrightarrow 1 : \Phi \lor \textcolor{red}{2} : \Phi \lor 3 : \Phi \\
		& 2 : \Diamond \Phi \Leftrightarrow \lor_{w_{2} \in 2+A} w_{2} :\Phi
		\Leftrightarrow \textcolor{red}{2} : \Phi \lor 3 : \Phi \lor 4 : \Phi
		\label{eq:k0sym 2}
	\end{aligned}
	\end{equation}
	Suppose $k : \Phi$ holds, where $k=2$. From relations (\ref{eq:k0sym 2}), we have $0 : \Diamond \Phi$, $1 : \Diamond \Phi$, and $2 : \Diamond \Phi$. From relation (\ref{eq:k0sym 1}), we have $0 : \Box \Diamond \Phi$. $0 : \Diamond \Phi$, $1 : \Diamond \Phi$, and $2 : \Diamond \Phi$ follow from $k \in 0+A$, $k \in 1+A$, and $k \in 2+A$, which coincides with $k$-symmetry.
    
	See Figure \ref{fig:kzerosym}. From $k$-symmetry, for all $a \in A$, the transition $0 \rightarrow a$ (solid edges) can be followed by the transition $a \rightarrow a+(k-a)=k$ (dashed edges). This reachability results in $M \vDash w+k : \Phi \Rightarrow M \vDash w : \Box \Diamond \Phi$. The illustration of $M \vDash w : \Diamond \Box \Phi \Rightarrow M \vDash w+k : \Phi$ is the dual.

\begin{figure}[t]
\centering
\begin{tikzpicture}[
    every node/.style={circle, draw, minimum size=5mm},
    every edge/.style={draw, ->}
]

	\node (0) at (0,0) {0};
    	\node (1) at (2,1) {1};
	\node (2) at (4,0) {2};
	
    	\path (0) edge (1);
    	\path [dashed] (0) edge[bend left=10] (2);
    	\path (0) edge[bend right=10] (2);
    	\path [dashed] (1) edge (2);
    	\path (0) edge[loop above] (0);
    	\path (2) edge[dashed, loop right] (2);

\end{tikzpicture}
\caption{$k$-symmetry with $W=\mathbb{Z}/6\mathbb{Z}$, $k=2$, and $A=\{0,1,2\}$}
\label{fig:kzerosym}
\end{figure}
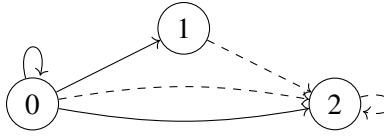
\end{example}

\begin{proof}[Proof of Lemma \ref{lem:leq}]
    We write $M$ for $(\mathbb{Z}/r\mathbb{Z}, R, V)$.
	\begin{enumerate}
		\item We prove that $M \vDash w+k : \Phi$ implies $M \vDash w : \Box \Diamond \Phi$. Suppose $M \vDash w+k : \Phi$. Take an arbitrary $a \in A$. From $M \vDash w+k : \Phi$, we have $M \vDash (w+a)+(k-a) : \Phi$. From $k-a \in A$ ($k$-symmetry), we have $M \vDash w+a : \Diamond \Phi$. Since $a \in A$ is arbitrary, we have $M \vDash w : \Box \Diamond \Phi$.
		\item We prove that $M \vDash w : \Diamond \Box \Phi$ implies $M \vDash w+k : \Phi$. Suppose $M \vDash w : \Diamond \Box \Phi$. From the definition of $\Diamond$, there exists some $a \in A$ such that $M \vDash w+a : \Box \Phi$. From $k-a \in A$ ($k$-symmetry), we have $M \vDash (w+a)+(k-a) : \Phi$, and thus $M \vDash w+k : \Phi$.
	\end{enumerate}
\end{proof}

	Now we are ready to prove Theorem \ref{th:calcu reduc pro ag}.
    
\begin{proof}[Proof of Theorem \ref{th:calcu reduc pro ag}]
	In Lemma \ref{lem:leq}, replace $\Phi$ with $\Box \Phi$ in the statement that $M \vDash w+k : \Phi$ implies $M \vDash w : \Box \Diamond \Phi$. Then, $M \vDash w+k : \Box \Phi$ implies $M \vDash w : \Box \Diamond \Box \Phi$. 
    
    Conversely, suppose $M \vDash w : \Box \Diamond \Box \Phi$. Take an arbitrary $a \in A$. Then, we have $M \vDash w+a : \Diamond \Box \Phi$. In Lemma \ref{lem:leq}, replace $w$ with $w+a$ in the statement that $M \vDash w : \Diamond \Box \Phi$ implies $M \vDash w+k : \Phi$. Then, $M \vDash w+a : \Diamond \Box \Phi$ implies $M \vDash (w+a)+k : \Phi$. Since $M \vDash w+a : \Diamond \Box \Phi$ is assumed, we have $M \vDash (w+a)+k : \Phi$, and thus $M \vDash (w+k)+a : \Phi$. Since $a \in A$ is arbitrary, we have $M \vDash w+k : \Box \Phi$.
\end{proof}

	If we set $k=0$, then the accessibility relation $R$ is symmetric. According to Exercise 3.1.1 in \cite{Blackburn}, for any Kripke frame $(W, R)$, the relation $R$ is symmetric if and only if for all valuations $V$ and $w \in W$, $M \vDash w : p$ implies $M \vDash w : \Box \Diamond p$. This implication relies on the standard assumption that formulas are evaluated at a fixed state. In contrast, the implication proved in Lemma \ref{lem:leq} involves an explicit shift of the evaluation state, and thereby cannot be derived from the standard symmetric-frame argument.

\section{Proof of Theorem \ref{th:P and DN is D}}\label{appendix: arrow}

	In this section, we prove Theorem \ref{th:P and DN is D}, a general criterion for dictatorship. We restate this theorem. Under an impossibility frame $(W, R)$, $\Gamma$ is logically interconnected enough to entail dictatorship. Recall that $F$ is dictatorial if there exists some individual $i_{0} \in N$ such that for all $f \in \mathcal{J}^{N}$, $F(f)=f(i_{0})$.

\begin{theorem}\label{th:P and DN is D app}
	Suppose $(W, R)$ is an impossibility frame under $\Gamma$. If $F:\mathcal{J}^{N} \to \mathcal{J}$ satisfies Unanimity and Positive-Negative Neutrality, then it is dictatorial.
\end{theorem}

	From now on, we assume all conditions in Theorem \ref{th:P and DN is D app}. To prove this theorem, we adopt a strategy similar to that in \cite{List}. First, by using the logical interconnections among propositions, we propagate decisiveness from one proposition to another (Proposition \ref{prop:sys and mono}). In other words, if a group is decisive over a certain proposition, the group is decisive over propositions logically connected with the original one. Second, we show that the family of decisive groups is closed under intersection (Proposition \ref{prop:intersection}). That is, if two groups are decisive, their intersection is also decisive. Third, by using this Intersection property, we prove that a single individual is decisive (Proposition \ref{prop:ultra prin}). Hence, the aggregation rule is dictatorial.

    \smallskip

Before the proof of Theorem \ref{th:P and DN is D app}, we prepare Lemma \ref{lem:cons extend} and \ref{lem:PN imp I}. We use Lemma \ref{lem:cons extend} to construct profiles of judgments in the proof of Lemma \ref{lem:PN imp I}, Proposition \ref{prop:sys and mono}, and Proposition \ref{prop:intersection}.

\begin{lemma}\label{lem:cons extend}
	For all $\Gamma_{0} \subseteq \Gamma$, if $\Gamma_{0}$ is consistent, then there exists some $\Gamma_{1} \in \mathcal{J}$ such that $\Gamma_{0} \subseteq \Gamma_{1}$.
\end{lemma}

\begin{proof}
	Since $\Gamma_{0}$ is consistent, there exists some valuation $V_{0}$ such that for all $\Phi \in \Gamma_{0}$, $M_{0}=(W, R, V_{0}) \vDash \ini : \Phi$. Set $\Gamma_{1}:=\{\Phi \in \Gamma \mid M_{0}=(W, R, V_{0}) \vDash \ini : \Phi\}$. Trivially, we have $\Gamma_{0} \subseteq \Gamma_{1}$ and $\Gamma_{1}$ is consistent. It remains to prove that $\Gamma_{1}$ is complete. Take $\Phi \in \Gamma$. If $M_{0} \vDash \ini : \Phi$, then $\Phi \in \Gamma_{1}$. If $M_{0} \nvDash \ini : \Phi$, then $M_{0} \vDash \ini : \lnot \Phi$, and thus $\lnot \Phi \in \Gamma_{1}$.
\end{proof}

\begin{lemma}\label{lem:PN imp I}
	Positive-Negative Neutrality implies Independence.
\end{lemma}

\begin{proof}
	Take $f,g \in \mathcal{J}^{N}$ and $\Phi \in \Gamma$ with $f^{-1}(\Phi)=g^{-1}(\Phi)=A$. Since both $\Phi$ and $\lnot \Phi$ are consistent, from Lemma \ref{lem:cons extend}, there exist some $\Gamma_{1}, \Gamma_{2} \in \mathcal{J}$ such that $\Phi \in \Gamma_{1}$ and $\lnot \Phi \in \Gamma_{2}$. We define $h \in \mathcal{J}^{N}$ as follows.
	\begin{align*}
		h(i):=\begin{cases}
				\Gamma_{2} & \quad (i \in A) \\
				\Gamma_{1} & \quad (i \in N \setminus A)
				\end{cases}
	\end{align*}
	Then we have $f^{-1}(\Phi)=A=h^{-1}(\lnot \Phi)$ and $g^{-1}(\Phi)=A=h^{-1}(\lnot \Phi)$. From Positive-Negative Neutrality, we have $\Phi \in F(f) \leftrightarrow \lnot \Phi \in F(h) \leftrightarrow \Phi \in F(g)$.
\end{proof}

	From Independence, for all $\Phi \in \Gamma$, we have the following equation.

\begin{align*}
	\{A \subseteq N \mid \forall f \in \mathcal{J}^{N}, f^{-1}(\Phi)=A \rightarrow \Phi \in F(f)\}
	=\{A \subseteq N \mid \exists f \in \mathcal{J}^{N}, f^{-1}(\Phi)=A \land \Phi \in F(f)\}
\end{align*}

	We write the above set as $\mathcal{U}_{\Phi}$. $\mathcal{U}_{\Phi}$ are called sets of decisive coalitions. From Positive-Negative Neutrality, for all $\Phi \in \Gamma$, we have $\mathcal{U}_{\Phi}=\mathcal{U}_{\lnot \Phi}$.

    Now we begin to prove Theorem \ref{th:P and DN is D app}. First, we prove Proposition \ref{prop:sys and mono}.

\begin{prop}\label{prop:sys and mono}
	For all $\Phi, \Psi \in \Gamma$ with no negation operator satisfying $\Phi <_{0} \Psi$, $A \in \mathcal{U}_{\Phi}$ and $A \subseteq B \subseteq N$ imply $B \in \mathcal{U}_{\Psi}$.
\end{prop}

\begin{proof}
	From $\Phi <_{0} \Psi$, take $\Gamma_{0} \subseteq \Gamma$ such that $\Gamma_{0} \cup \{\Phi, \lnot \Psi\}$ is minimally inconsistent and $\Gamma_{0} \cup \{\lnot \Phi, \Psi\}$ is consistent. From minimal inconsistency, $\Gamma_{0} \cup \{\Phi, \Psi\}$ and $\Gamma_{0} \cup \{\lnot \Phi, \lnot \Psi \}$ are consistent. By using Lemma \ref{lem:cons extend}, take $\Gamma_{1}, \Gamma_{2}, \Gamma_{3} \in \mathcal{J}$ such that $\Gamma_{0} \cup \{\Phi, \Psi\} \subseteq \Gamma_{1}$ and $\Gamma_{0} \cup \{\lnot \Phi, \lnot \Psi \} \subseteq \Gamma_{2}$ and $\Gamma_{0} \cup \{\lnot \Phi, \Psi\} \subseteq \Gamma_{3}$.

	We define $h \in \mathcal{J}^{N}$ as follows.
\begin{align*}
	h(i):=\begin{cases}
			\Gamma_{1} & \quad (i \in A) \\
			\Gamma_{3} & \quad (i \in B \setminus A) \\
			\Gamma_{2} & \quad (i \in N \setminus B)
			\end{cases}
\end{align*}

	We prove $\Gamma_{0} \cup \{\Phi\} \subseteq F(h)$. From Unanimity, we have $\Gamma_{0} \subseteq F(h)$. From $A \in \mathcal{U}_{\Phi}$ and $h^{-1}(\Phi)=A$, we have $\Phi \in F(h)$.
    
    Since we have $\Gamma_{0} \cup \{\Phi\} \subseteq F(h)$ and $\Gamma_{0} \cup \{\Phi, \lnot \Psi\}$ is inconsistent, we have $\Psi \in F(h)$. Together with $h^{-1}(\Psi)=B$, we have $B \in \mathcal{U}_{\Psi}$.
\end{proof}

	From Proposition \ref{prop:sys and mono} and Strong Path-Connectedness and Positive-Negative Neutrality, for all $\Phi, \Psi \in \Gamma$, we have $\mathcal{U}_{\Phi}=\mathcal{U}_{\Psi}$. Therefore, we set $\mathcal{U}:=\mathcal{U}_{\Phi}$. Also, for all $A \subseteq B \subseteq N$, $A \in \mathcal{U}$ implies $B \in \mathcal{U}$, which is called Monotonicity.

    Second, we prove Proposition \ref{prop:intersection}.

\begin{prop}\label{prop:intersection}
	For all $A, B \subseteq N$, $A \in \mathcal{U}$ and $B \in \mathcal{U}$ imply $A \cap B \in \mathcal{U}$.
\end{prop}

\begin{proof}
	From minimal connectedness, take $\Gamma_{0} \subseteq \Gamma$ with $|\Gamma_{0}| \ge 3$ such that $\Gamma_{0}$ is minimally inconsistent. Take distinct $\Phi_{1}, \Phi_{2}, \Phi_{3} \in \Gamma_{0}$. From minimal inconsistency, $(\Gamma_{0} \setminus \{\Phi_{i}\}) \cup \{\lnot \Phi_{i}\}$ is consistent for all $i=1,2,3$. By using Lemma \ref{lem:cons extend}, for all $i=1,2,3$, take $\Gamma_{i} \in \mathcal{J}$ such that $(\Gamma_{0} \setminus \{\Phi_{i}\}) \cup \{\lnot \Phi_{i}\} \subseteq \Gamma_{i}$.

	We define $h \in \mathcal{J}^{N}$ as follows.
\begin{align*}
	h(i):=\begin{cases}
			\Gamma_{1} & \quad (i \in A \cap B) \\
			\Gamma_{2} & \quad (i \in A \setminus B) \\
			\Gamma_{3} & \quad (i \in N \setminus A)
		\end{cases}
\end{align*}

We prove $\Gamma_{0} \setminus \{\Phi_{1}\} \subseteq F(h)$ by the following arguments.
\begin{enumerate}
	\item From Unanimity, we have $\Gamma_{0} \setminus \{\Phi_{1},\Phi_{2},\Phi_{3}\} \subseteq F(h)$. 
    \item From $A \in \mathcal{U}=\mathcal{U}_{\Phi_{3}}$ and $A=h^{-1}(\Phi_{3})$, we have $\Phi_{3} \in F(h)$. 
    \item From $B \in \mathcal{U}=\mathcal{U}_{\Phi_{2}}$ and $B \subseteq (A \cap B) \cup (N \setminus A)$ and Monotonicity, we have $(A \cap B) \cup (N \setminus A) \in \mathcal{U}$. Together with $(A \cap B) \cup (N \setminus A)=h^{-1}(\Phi_{2})$, we have $\Phi_{2} \in F(h)$. 
\end{enumerate}

    Therefore, we have $\Gamma_{0} \setminus \{\Phi_{1}\} \subseteq F(h)$. Together with the inconsistency of $\Gamma_{0}$, we have $\lnot \Phi_{1} \in F(h)$. Also, from $\lnot \Phi_{1} \in \Gamma_{1}$ and $\Phi_{1} \in \Gamma_{2}$ and $\Phi_{1} \in \Gamma_{3}$, we have $A \cap B =h^{-1}(\lnot \Phi_{1})$. From $\lnot \Phi_{1} \in F(h)$ and $A \cap B=h^{-1}(\lnot \Phi_{1})$, we have $A \cap B \in \mathcal{U}_{\lnot \Phi_{1}}=\mathcal{U}$.
\end{proof}

	From Proposition \ref{prop:sys and mono} and \ref{prop:intersection}, $\mathcal{U} \subseteq 2^N$ is an ultrafilter on $N$. That is, it satisfies the following four conditions.

\begin{enumerate}
	\item From Unanimity, we have $N \in \mathcal{U}$ and $\varnothing \notin \mathcal{U}$.
	\item From Monotonicity, for all $A \subseteq B \subseteq N$, $A \in \mathcal{U}$ implies $B \in \mathcal{U}$.
	\item From Proposition \ref{prop:intersection}, for all $A, B \subseteq N$, $A \in \mathcal{U}$ and $B \in \mathcal{U}$ imply $A \cap B \in \mathcal{U}$, which is called Intersection.
	\item From Positive-Negative Neutrality, we have $\mathcal{U}_{\Phi}=\mathcal{U}_{\lnot \Phi}$. Therefore, for all $A \subseteq N$, we have $A \in \mathcal{U}$ or $N \setminus A \in \mathcal{U}$.
\end{enumerate}

    Ultrafilters have been treated in social choice theory. Fishburn \cite{Fishburn1970} first used ultrafilters to construct non-dictatorial arrowian preference aggregations in infinite population. Kirman and Sondermann \cite{Kirman1972} then established correspondence between arrowian preference aggregations and ultrafilters. In judgment aggregation, Herzberg and Eckert \cite{Herzberg2012} proved that under some conditions, dictatorship is inevitable in infinite population even if ultrafilter aggregations are admitted.

	Third, we prove Proposition \ref{prop:ultra prin}. By this proposition, $F$ is dictatorial.

\begin{prop}\label{prop:ultra prin}
	There exists some $i_{0} \in N$ such that $\mathcal{U}=\{A \subseteq N \mid i_{0} \in A\}$.
\end{prop}

\begin{proof}
    We prove that there exists some $i_{0} \in N$ such that $\{A \subseteq N \mid i_{0} \in A\} \subseteq \mathcal{U}$. Assume that for all $i \in N$, we have $N \setminus \{i\} \in \mathcal{U}$. Since $N$ is finite, from Intersection, we have $\varnothing=\cap_{i \in N} N \setminus \{i\} \in \mathcal{U}$, which contradicts $\varnothing \notin \mathcal{U}$. Therefore, there exist some $i_{0} \in N$ such that $N \setminus \{i_{0}\} \notin \mathcal{U}$. Together with Positive-Negative Neutrality, we have $\{i_{0}\} \in \mathcal{U}$. From Monotonicity, we have $\{A \subseteq N \mid i_{0} \in A\} \subseteq \mathcal{U}$. 

    Conversely, let $A \in \mathcal{U}$. From $\{i_{0}\} \in \mathcal{U}$ and Intersection, we have $A \cap \{i_{0}\} \in \mathcal{U}$. Together with $\varnothing \notin \mathcal{U}$, we have $A \cap \{i_{0}\} \neq \varnothing$, and thus $i_{0} \in A$.
\end{proof}

\section{Proofs omitted from Section \ref{sec:impossible agenda}}\label{sec:impossible agenda pro}

	In this section, we prove one lemma and two propositions which we do not prove in Section \ref{sec:impossible agenda}. Recall that proposition $P_{w}$ means $A+w \subseteq V$, where $V$ is the valuation of the Kripke model.
    
    The following proposition restates Proposition \ref{lem:plus minus con}, which corresponds to Contribution \ref{con2}.

\begin{prop}\label{lem:plus minus con app}
	Let $J_{+}, J_{-} \subseteq \mathbb{Z}/r\mathbb{Z}$. Then $\Gamma_{0}:=\{P_{w}\}_{w \in J_{+}} \sqcup \{\lnot P_{w_{0}}\}_{w_{0} \in J_{-}}$ is consistent if and only if for all $w_{0} \in J_{-}$, $A+w_{0} \not\subseteq \cup_{w \in J_{+}} A+w$.
\end{prop}

\begin{proof}
	Suppose $\Gamma_{0}$ is consistent. Then there exists some valuation $V \subseteq W$ such that for all $w \in J_{+}$, $P_{w}$ is true, and for all $w_{0} \in J_{-}$, $P_{w_{0}}$ is false. Therefore, for all $w \in J_{+}$, we have $A+w \subseteq V$, and thus $\cup_{w \in J_{+}} A+w \subseteq V$. Also, take an arbitrary $w_{0} \in J_{-}$. From the condition on $V$, we have $A+w_{0} \not\subseteq V$. Together with $\cup_{w \in J_{+}} A+w \subseteq V$, we have $A+w_{0} \not\subseteq \cup_{w \in J_{+}} A+w$.
	
	Conversely, suppose for all $w_{0} \in J_{-}$, we have $A+w_{0} \not\subseteq \cup_{w \in J_{+}} A+w$. Set $V:=\cup_{w \in J_{+}} A+w$. For all $w \in J_{+}$, we have $A+w \subseteq V$, and thus $P_{w}$ is true. For all $w_{0} \in J_{-}$, we have $A+w_{0} \not\subseteq V$, and thus $\lnot P_{w_{0}}$ is true.
\end{proof}

The following lemma restates Lemma \ref{lem:min cov good}.

\begin{lemma}\label{lem:min cov good app}
If $(w_{0}, w_{1},S_{0})$ is a pointed minimal cover, then the following statements hold.
\begin{enumerate}
	\item \label{lem:min cov good app 1} $\{\lnot P_{w_{0}}\} \cup \{P_{s}\}_{s \in S_{0}}$ is minimally inconsistent.
	\item $P_{w_{1}} <_{0} P_{w_{0}}$.
\end{enumerate}
\end{lemma}

\begin{proof}
\begin{enumerate}
	\item From $A+w_{0} \subseteq \cup_{s \in S_{0}} A+s$ and Proposition \ref{lem:plus minus con app}, $\{\lnot P_{w_{0}}\} \cup \{P_{s}\}_{s \in S_{0}}$ is inconsistent. Conversely, let $\Gamma_{0} \subseteq \{\lnot P_{w_{0}}\} \cup \{P_{s}\}_{s \in S_{0}}$ be inconsistent. We prove $\lnot P_{w_{0}} \in \Gamma_{0}$ by contradiction. Assume $\Gamma_{0} \subseteq \{P_{s}\}_{s \in S_{1}}$. If we set $V:=W$, all formulas in $\Gamma_{0}$ are true, which is a contradiction. Therefore, we have $\lnot P_{w_{0}} \in \Gamma_{0}$. Suppose $\Gamma_{0}=\{\lnot P_{w_{0}}\} \cup \{P_{s}\}_{s \in S_{1}}$, where $S_{1} \subseteq S_{0}$. Since $\Gamma_{0}$ is inconsistent, from Proposition \ref{lem:plus minus con app}, we have $A+w_{0} \subseteq \cup_{s \in S_{1}} A+s$. From Minimality, we have $S_{1}=S_{0}$, and thus $\Gamma_{0}=\{\lnot P_{w_{0}}\} \cup \{P_{s}\}_{s \in S_{0}}$.
	
	\item Take $\Gamma_{0}=\{P_{s}\}_{s \in S_{0} \setminus \{w_{1}\}}$ in the definition of $<_{0}$. From Irreducibility and Proposition \ref{lem:plus minus con app}, $\{\lnot P_{w_{1}}\} \cup \{P_{w}\}_{w \in (S_{0} \setminus \{w_{1}\}) \cup \{w_{0}\}}=\Gamma_{0} \cup \{\lnot P_{w_{1}}, P_{w_{0}}\}$ is consistent. Also, from Lemma \ref{lem:min cov good app} (\ref{lem:min cov good app 1}), $\Gamma_{0} \cup \{P_{w_{1}}, \lnot P_{w_{0}}\}$ is minimally inconsistent. Therefore, we have $P_{w_{1}} <_{0} P_{w_{0}}$.
\end{enumerate}
\end{proof}

The following proposition restates Proposition \ref{prop:ex of minimal cover}, which constructs pointed minimal covers. Recall that for any integers $a,b \in \mathbb{Z}$ with $a \le b$, let $[a,b]=\{\overline{a},\overline{a+1},\cdots,\overline{b}\} \subseteq \mathbb{Z}/r\mathbb{Z}$. In the proof, we write $\overline{0}$ and $\overline{k}$ instead of abbreviating it as $0$ and $k$.

\begin{prop}\label{prop:ex of minimal cover app}
	For all $w \in \mathbb{Z}/r\mathbb{Z}$, there exists a subset $S_{0} \subseteq \mathbb{Z}/r\mathbb{Z}$ with $|S_{0}| \ge 2$ and $w+\overline{k} \in S_{0}$ such that $(w, w+\overline{k}, S_{0})$ is a pointed minimal cover.
\end{prop}

\begin{proof}
	It is enough to prove Proposition \ref{prop:ex of minimal cover app} in the case $w=\overline{0}$.
    
	We have $A+\overline{0}=A \subseteq A+\overline{k} \cup (\cup_{a \in A \setminus \{\overline{k}\}} A+(a-\overline{k}))$. Suppose $a \in A$. If $a=\overline{k}$, then we have $\overline{k}=\overline{0}+\overline{k} \in A+\overline{k}$. If $a \in A \setminus \{\overline{k}\}$, then we have $a=\overline{k}+(a-\overline{k}) \in A+(a-\overline{k})$.
	
	Since the above cover is a finite cover, there exists some $S_{0} \subseteq \{\overline{k}\} \cup \{a-\overline{k} \mid a \in A \setminus \{\overline{k}\}\}$ such that $\cup_{s \in S_{0}} A+s$ is a minimal cover of $A$. The following arguments show that the triple $(\overline{0},\overline{k},S_{0})$ is a pointed minimal cover.
	\begin{enumerate}
		\item We prove $\overline{k} \in S_{0}$. From $\overline{k} \in A \subseteq \cup_{s \in S_{0}} A+s$, there exists some $s_{0} \in S_{0}$ such that $\overline{k} \in A+s_{0}$. We prove $s_{0}=\overline{k}$ by contradiction. Assume $s_{0} \in S_{0} \setminus \{\overline{k}\}$. From $S_{0} \subseteq \{\overline{k}\} \cup \{a-\overline{k} \mid a \in A \setminus \{\overline{k}\}\}$, $s_{0}$ is of the form $a-\overline{k}$ with $a \in A \setminus \{\overline{k}\}$. Thus, we have the following relation in $\mathbb{Z}/r\mathbb{Z}$.
        \begin{align*}
            \overline{k} &\in A+s_{0} \quad (\text{Choice of $s_{0}$}) \\
            &=A+(a-\overline{k}) \quad (\text{Assumption $s_{0}=a-\overline{k}$}) \\
            &\subseteq [0, k]+[0, k-1]-\overline{k} \quad (A \subseteq [0, k], a \in A \setminus \{\overline{k}\} \subseteq [0, k-1]) \\
            &=[-k, k-1]
        \end{align*}
        Therefore, we have $\overline{0} \in [-2 \cdot k, -1]$ in $\mathbb{Z}/r\mathbb{Z}$. However, from $0 <k < \dfrac{r}{3}$, we have $-r<-2 \cdot k<-1<0$, which contradicts $\overline{0} \in [-2 \cdot k, -1]$ in $\mathbb{Z}/r\mathbb{Z}$. Therefore, we have $s_{0}=\overline{k}$, and thus $\overline{k} \in S_{0}$.
        
		\item We prove $|S_{0}| \ge 2$. It is enough to show the existence of $s_{0} \in S_{0} \setminus \{\overline{k}\}$. From $\overline{0} \in A \subseteq \cup_{s \in S_{0}} A+s$, there exists some $s_{0} \in S_{0}$ such that $\overline{0} \in A+s_{0}$. We prove $s_{0} \in S_{0} \setminus \{\overline{k}\}$ by contradiction. Assume $s_{0}=\overline{k}$. From $\overline{0} \in A+s_{0}$, $A \subseteq [0, k]$, and $s_{0}=\overline{k}$, we have $\overline{0} \in A+\overline{k} \subseteq [k, 2 \cdot k]$ in $\mathbb{Z}/r\mathbb{Z}$. However, from $0 <k < \dfrac{r}{3}$, we have $0<k<2 \cdot k<r$, which contradicts $\overline{0} \in [k, 2 \cdot k]$ in $\mathbb{Z}/r\mathbb{Z}$. Therefore, we have $s_{0} \in S_{0} \setminus \{\overline{k}\}$, and thus $|S_{0}| \ge 2$.
        
		\item Minimality: From the definition of $S_{0}$, $\cup_{s \in S_{0}} A+s$ is a minimal cover of $A+\overline{0}=A$.
        
		\item Irreducibility: It is enough to show $2 \cdot \overline{k} \in A+\overline{k} \setminus (\cup_{w \in (S_{0} \setminus \{\overline{k}\}) \cup \{\overline{0}\}} A+w)$. First, we have $2 \cdot \overline{k}=\overline{k}+\overline{k} \in A+\overline{k}$. 
		
		Second, we prove $2 \cdot \overline{k} \notin A+s$ for all $s \in S_{0} \setminus \{\overline{k}\}$ by contradiction. Assume $2 \cdot \overline{k} \in A+s$. From $s \in S_{0} \setminus \{\overline{k}\} \subseteq \{a-\overline{k} \mid a \in A \setminus \{\overline{k}\}\}$, $s$ is of the form $a-\overline{k}$ with $a \in A \setminus \{\overline{k}\}$.  Together with $A \subseteq [0, k]$, we have the following relation in $\mathbb{Z}/r\mathbb{Z}$.
		\begin{align*}
			2 \cdot \overline{k} &\in A+s \quad (\text{Assumption}) \\
            &=A+a-\overline{k} \\
			&\subseteq A+A-\overline{k} \quad (a \in A \setminus \{\overline{k}\}) \\
            &\subseteq [0, k]+[0, k]-\overline{k} \quad (A \subseteq [0,k]) \\
			&\subseteq [-k, k]
		\end{align*}
		Therefore, we have $\overline{0} \in [-3 \cdot k, -k]$ in $\mathbb{Z}/r\mathbb{Z}$. However, from $0<k < \dfrac{r}{3}$, we have $-r<-3 \cdot k<-k<0$, which contradicts $\overline{0} \in [-3 \cdot k, -k]$ in $\mathbb{Z}/r\mathbb{Z}$. Therefore, we have $2 \cdot \overline{k} \notin A+s$.
		
		Third, we prove $2 \cdot \overline{k} \notin A+0=A$ by contradiction. Assume $2 \cdot \overline{k} \in A+\overline{0}=A$. Together with $A \subseteq [0, k]$, we have $2 \cdot \overline{k} \in [0, k]$, and thus $\overline{0} \in [-2 \cdot k, -k]$ in $\mathbb{Z}/r\mathbb{Z}$. However, from $0<k < \dfrac{r}{3}$, we have $-r<-2 \cdot k<-k<0$, which contradicts $\overline{0} \in [-2 \cdot k, -k]$ in $\mathbb{Z}/r\mathbb{Z}$. Therefore, we have $2 \cdot \overline{k} \notin A+0=A$.
	\end{enumerate}
\end{proof}
\end{appendices}

\bibliographystyle{plain}
\bibliography{ref}

\end{document}